\documentclass[conference]{IEEEtran}
\pdfoutput=1
\usepackage{subfigure}
\usepackage{multirow}
\usepackage{graphicx}
\usepackage{amsmath}
\usepackage{color}
\usepackage{amsfonts}
\usepackage{amssymb}
\usepackage{amsthm}
\usepackage{algorithm}
\usepackage{algorithmic}

\newtheorem{lem}{Lemma}
\newtheorem{prop}{Proposition}

\newcommand{\argmax}{\mathop{\text{argmax}}}
\newcommand{\argmin}{\mathop{\text{argmin}}}

\newcommand{\vv}{{\bf v}}
\newcommand{\vw}{{\bf w}}
\newcommand{\vx}{{\bf x}}
\newcommand{\vy}{{\bf y}}
\newcommand{\vz}{{\bf z}}
\newcommand{\vp}{{\bf p}}

\newcommand{\mj}{{\bf J}}
\newcommand{\mh}{{\bf H}}
\newcommand{\mi}{{\bf I}}

\newcommand{\mx}{{\bf X}}
\newcommand{\my}{{\bf Y}}

\newcommand{\mw}{{\bf W}}

\newcommand{\mje}{\hat{\bf J}}
\newcommand{\vze}{\hat{\bf z}}

\newcommand{\tj}{\widetilde{J}}
\newcommand{\tmj}{\widetilde{\mj}}
\newcommand{\tvw}{\widetilde{\vw}}

\newcommand{\E}{{\mathbb E}}

\newcommand{\Define}{\triangleq}

\begin{document}
\twocolumn

\title{{\huge 
Channel Hardening-Exploiting Message Passing (CHEMP) Receiver
in Large-Scale MIMO Systems
}}

\author{T. Lakshmi Narasimhan and A. Chockalingam \\
Department of ECE, Indian Institute of Science, Bangalore}

\IEEEaftertitletext{\vspace{-0.6\baselineskip}}
\maketitle
\begin{abstract}
In this paper, we propose a multiple-input multiple-output (MIMO) receiver
algorithm that exploits {\em channel hardening} that occurs in large MIMO
channels. Channel hardening refers to the phenomenon where the off-diagonal
terms of the ${\bf H}^H{\bf H}$ matrix become increasingly weaker compared
to the diagonal terms as the size of the channel gain matrix ${\bf H}$
increases. Specifically, we propose a message passing detection (MPD) 
algorithm which works with the real-valued matched filtered received vector 
(whose signal term becomes ${\bf H}^T{\bf H}{\bf x}$, where ${\bf x}$ is the
transmitted vector), and uses a Gaussian approximation on the off-diagonal
terms of the  ${\bf H}^T{\bf H}$ matrix. We also propose a simple estimation
scheme which directly obtains an estimate of ${\bf H}^T{\bf H}$ (instead of 
an estimate of  ${\bf H}$), which is used as an effective channel estimate 
in the MPD algorithm. We refer to this receiver as the {\em channel 
hardening-exploiting message passing (CHEMP)} receiver. The proposed CHEMP 
receiver achieves very good performance in large-scale MIMO systems (e.g., in
systems with 16 to 128 uplink users and 128 base station antennas). For the 
considered large MIMO settings, the complexity of the proposed MPD algorithm 
is almost the same as or less than that of the minimum mean square error (MMSE) 
detection. This is because the MPD algorithm does not need a matrix inversion. 
It also achieves a significantly better performance compared to MMSE and other 
message passing detection algorithms using MMSE estimate of ${\bf H}$. We also 
present a convergence analysis of the proposed MPD algorithm. Further, we 
design optimized irregular low density parity check (LDPC) codes specific to 
the considered large MIMO channel and the CHEMP receiver through EXIT chart 
matching. The LDPC codes thus obtained achieve improved coded bit error rate 
performance compared to off-the-shelf irregular LDPC codes.
\end{abstract}

\section{Introduction}
\label{sec1}
Wireless communication systems using multiple-input multiple-output (MIMO)
configurations with a large number of antennas have attracted a lot of
research attention \cite{lmimo1},\cite{lmimo2},\cite{lmimo3},\cite{scale}. 
These systems can achieve high spectral and power efficiencies. An emerging 
architecture for large-scale multiuser MIMO communications is one where each 
base station (BS) is equipped with a large number of antennas and the user 
terminals are equipped with one antenna each. A key requirement on the uplink 
(user terminal to BS link) in such large-scale MIMO systems is to achieve 
reduced channel estimation, detection and decoding complexities at the BS 
receiver to enable practical implementation, while maintaining good 
performance. When the number of BS antennas is much larger than the number 
of uplink users (i.e., low system loading factors), linear detectors like 
the minimum mean square error (MMSE) detector are good in terms of both 
complexity and performance \cite{mmse1}. In the recent years, several low 
complexity detection algorithms which achieve near-optimal performance in 
large dimensions using complexities comparable to that of MMSE detection 
have been proposed \cite{lmimo1},\cite{lmimo2},\cite{las1}-\cite{heuris1}.
These algorithms are based on local search \big(e.g., likelihood ascent 
search (LAS) algorithm and variants in 
\cite{lmimo1},\cite{lmimo2},\cite{las1},\cite{las2}\big), 
meta-heuristics \big(e.g., reactive tabu search (RTS) and variants in 
\cite{rts1},\cite{rts2}\big), message passing techniques \big(e.g.,
belief propagation (BP) based algorithms in \cite{jstsp},\cite{gta}\big),
lattice reduction techniques \big(e.g., lattice reduction (LR) aided 
detectors in \cite{lattice1},\cite{lattice2}\big), and Monte-Carlo 
sampling techniques \big(e.g., Markov chain Monte Carlo (MCMC) algorithms 
in \cite{mcmc1}\big). Issues related channel estimation and low density 
parity check codes for large-scale MIMO systems are also being addressed 
\cite{cest1},\cite{comm}. 

Message passing on graphical models is a promising low-complexity 
high-performance approach for signal processing in large dimensions 
\cite{frey}. Decoding of turbo codes and LDPC codes, and 
equalization/detection \cite{bp1}-\cite{bickson} are popular 
examples of the use of message passing algorithms in communications. 
In \cite{jstsp}, a MIMO detection algorithm based on approximate 
message passing on a factor graph is presented. The message passing
algorithm in \cite{gta} uses a different approach. It obtains a tree 
that approximates the fully-connected MIMO graph and performs message 
passing on this tree. 

In this this paper, we propose a promising low-complexity receiver 
for large-scale MIMO systems. The receiver is based on  message passing. 
The novelty in the proposed receiver lies in the exploitation of the  
`channel hardening' phenomenon that occurs in large MIMO channels
\cite{chard1},\cite{tse},\cite{chard2},\cite{chard3}.
Channel hardening refers to the phenomenon where the off-diagonal
terms of the ${\bf H}^T{\bf H}$ matrix become increasingly weaker compared
to the diagonal terms as the size of the channel gain matrix ${\bf H}$
increases. We exploit this for the purposes of detection and channel
estimation. The proposed receiver, referred to as the {\em channel 
hardening-exploiting message passing (CHEMP)} receiver, consists of two 
components; a message passing detection (MPD) algorithm and an estimation 
scheme to obtain an estimate of $\mh^T\mh$. The highlights of our 
contributions in this paper can be summarized as follows:
\begin{itemize}
\vspace{-3mm}
\item
proposal of the MPD algorithm which works with the real-valued matched 
filtered received vector, 
and uses a Gaussian approximation on the off-diagonal terms of the  
${\bf H}^T{\bf H}$ matrix.
\item 
proposal of a simple estimation scheme which directly obtains an 
estimate of ${\bf H}^T{\bf H}$ (instead of an estimate of  ${\bf H}$), 
which is used as an effective channel estimate in the MPD algorithm. 
\item
less than the MMSE detection complexity (because matrix inversion is 
not needed in the MPD algorithm).
\item significantly 
better performance compared to MMSE and other message passing detection 
algorithms which use MMSE estimate of ${\bf H}$. 
\item 
convergence analysis of the MPD algorithm which proves the existence
of a fixed point in the MPD algorithm.
\item
analysis of the mean square difference of the log-likelihood ratios (LLRs)
in the proposed receiver with perfect and estimated channel state information
(CSI).
\item 
design of optimized irregular LDPC codes specific to the considered large 
MIMO channel and the CHEMP receiver through EXIT chart matching. 
\end{itemize}

The rest of the paper is organized as follows. The system model and the
channel hardening phenomenon are described in Section \ref{sec2}. The 
proposed CHEMP receiver, and its performance and complexity are presented 
in Section \ref{sec3}. An analysis of the CHEMP receiver is presented in
Section \ref{sec4}.  Section \ref{subsec_new} presents an
extension to higher-order QAM. The design and performance of LDPC codes 
matched to the large MIMO channel and the CHEMP receiver are presented in 
Section \ref{sec5}. Conclusions are presented in Section \ref{sec6}.

\section{System Model}
\label{sec2}
Consider a large-scale multiuser MIMO system where $K$ uplink users, each
transmitting with a single antenna, communicate with a BS having a large
number of receive antennas. Let $N$ denote the number of BS antennas; $N$
is in the range of tens to hundreds. The ratio $\alpha=K/N$ is the system
loading factor. We consider $\alpha \leq 1$ (i.e., $K \leq N$). The system 
model is illustrated in Fig. \ref{mumimo}. Each user encodes a sequence of 
$k$ information bits to a sequence of $n$ coded symbols using an LDPC code 
of code rate $R=k/n$. The encoded bits are modulated and transmitted. 
 Let ${\mathbb A}$ denote the modulation alphabet. The 
transmission of one LDPC code block requires $n/(\log_2|{\mathbb A}|)$ 
channel uses. 

\begin{figure}
\includegraphics[width=3.35in,height=2.55in]{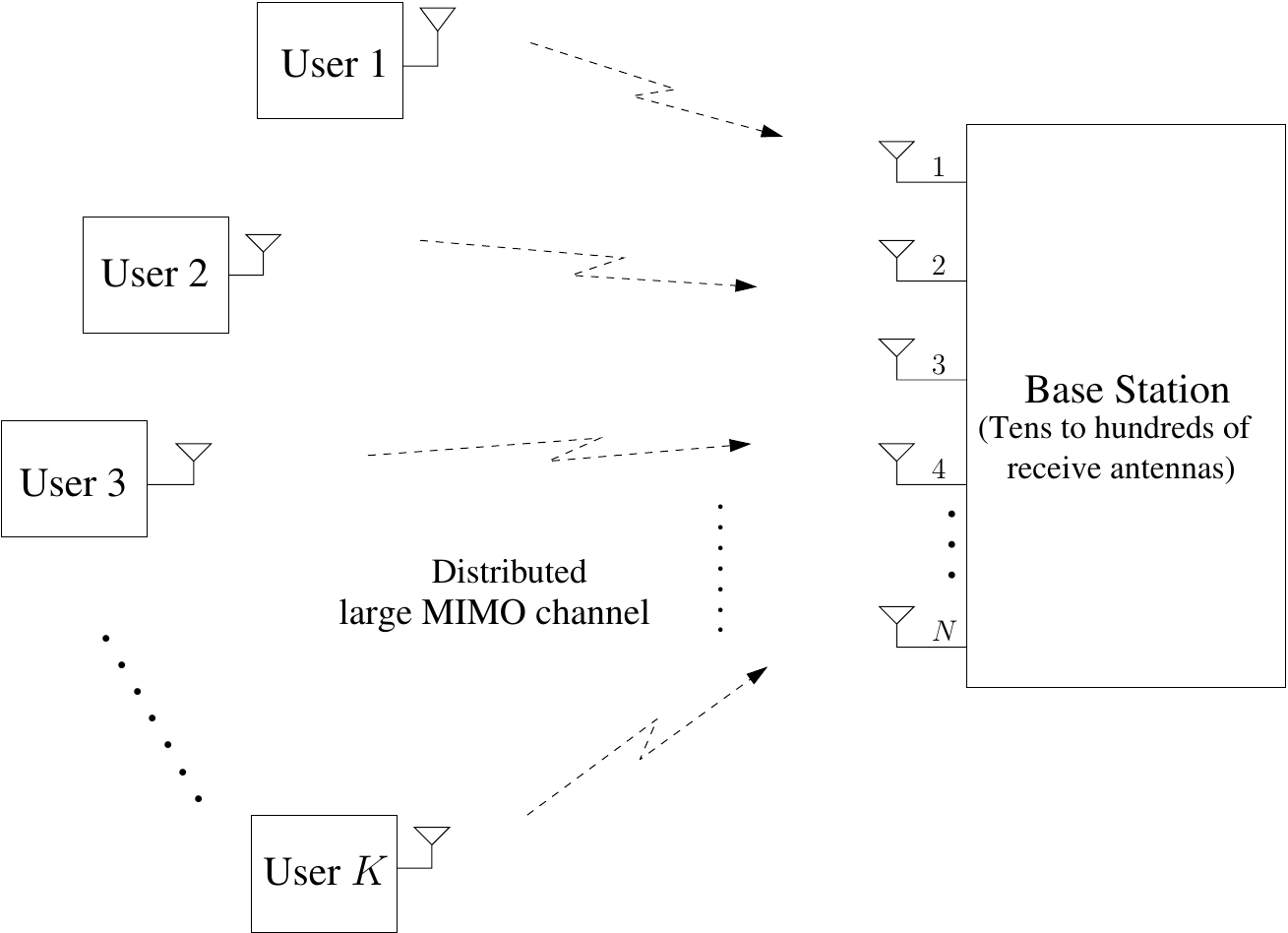} 
\caption{Large-scale multiuser MIMO system model on the uplink.}
\label{mumimo} 
\end{figure}

Let ${\bf H}_c^{(t)} \in \mathbb{C}^{N\times K}$ denote the channel gain 
matrix in the $t$th channel use and $H_{ij}^c$ denote the complex channel 
gain from the $j$th user to the $i$th BS antenna. The channel gains 
$H_{ij}^c$s are assumed to be independent Gaussian with zero mean and 
variance $\sigma_j^2$, such that $\sum_j \sigma_j^2=K$. The $\sigma_j^2$ 
models the imbalance in the received power from user $j$ due to path loss 
etc., and $\sigma_j^2=1$ corresponds to the case of perfect power control. 
Let ${\bf x}_c^{(t)} \in {\mathbb A}^K$ denote the modulated 
symbol vector transmitted in the $t$th channel use, where the $j$th element 
of ${\bf x}_c^{(t)}$ denotes the modulation symbol transmitted by the $j$th 
user. Assuming perfect synchronization, the received vector at the BS in the 
$t$th channel use, ${\bf y}_c^{(t)}$, is given by
\begin{eqnarray}
{\bf y}_c^{(t)} & = & {\bf H}_c^{(t)}{\bf x}_c^{(t)} + {\bf w}_c^{(t)},
\label{csys}
\end{eqnarray}
where ${\bf w}_c^{(t)}$ is the noise vector. Dropping the channel use index 
for convenience, (\ref{csys}) can be written in the real domain as
\begin{equation}
{\bf y} = {\bf Hx} + {\bf w},
\label{sys}
\end{equation}
where
\begin{equation}
{\bf H} \Define \left[\begin{array}{cc}\Re({\bf H}_c) \hspace{2mm} -\Im({\bf H}_c) \\
\Im({\bf H}_c)  \hspace{5mm} \Re({\bf H}_c) \end{array}\right], \nonumber
\end{equation}
\begin{equation}
{\bf y} \Define
\left[\begin{array}{c} \Re({\bf y}_c) \\ \Im({\bf y}_c) \end{array}\right], \,
{\bf x} \Define
\left[\begin{array}{c} \Re({\bf x}_c) \\ \Im({\bf x}_c) \end{array}\right], \, 
{\bf w} \Define
\left[\begin{array}{c} \Re({\bf w}_c) \\ \Im({\bf w}_c) \end{array}\right],
\nonumber
\end{equation}
$\Re(.)$ and $\Im(.)$ denote the real and imaginary parts, respectively.
Note that $\mh \in \mathbb{R}^{2N\times 2K}$, $\vy \in \mathbb{R}^{2N}$,
$\vw\in \mathbb{R}^{2N}$, and  $\vx\in{\mathbb R}^{2K}$. 
For a QAM alphabet ${\mathbb A}$, the elements of ${\bf x}$ will take 
values from the underlying PAM alphabet ${\mathbb B}$, i.e., 
${\bf x} \in {\mathbb B}^{2K}$.
The elements of $\vw$ are modeled as i.i.d. $\mathcal{N}(0,\sigma^2_n)$.
The average received SNR per receive antenna is given by
$\gamma=\frac{KE_s}{2\sigma_n^2}$, where $E_s$ is the average energy
of the transmitted symbols.
For the real-valued system model in (\ref{sys}), the maximum-likelihood (ML) 
detection rule is given by
\begin{equation}
\label{ml}
\hat{\vx}=\argmin_{\vx\in{\mathbb B}^{2K}} \ (\vy-\mh\vx)^T(\vy-\mh\vx).
\end{equation}
When the transmitted bits are equally likely, then the ML decision
rule is same as the maximum a posteriori  probability (MAP) decision 
rule, given by
\begin{eqnarray}
\label{map}
\hat{\vx}=\argmax_{\vx\in {\mathbb B}^{2K}} \ \Pr(\vx\mid\vy,\mh). 
\end{eqnarray}

The exact computation of (\ref{ml}) and (\ref{map}) requires exponential 
complexity in $K$. Message passing algorithms can provide approximate 
marginalization of the joint distribution in (\ref{map}) at low complexities. 
In Section \ref{sec3}, we propose such a message passing algorithm, whose 
novelty lies in exploiting the channel hardening phenomenon that happens in 
large MIMO channels. The channel hardening effect in large MIMO channels is 
described in the following subsection. 

\subsection{Channel hardening in large MIMO channels}
\label{sec2a}
Channel hardening refers to the phenomenon where the variance of the mutual
information of the MIMO channel grows very slowly relative to its mean or 
even shrink as the number of antennas grows \cite{chard1}. Consider a 
$n_r\times n_t$ MIMO channel. As $n_r$ and $n_t$ are increased keeping their 
ratio fixed, the distribution of the singular values of the MIMO channel 
matrix becomes less sensitive to the actual distribution of the entries of 
the channel matrix (as long as the entries are i.i.d.) \cite{tse}. This is 
a result of the Mar\v{c}enko-Pastur law \cite{chard2}, which states that if 
the entries of a $n_r\times n_t$ matrix $\mh$ are zero mean i.i.d. with 
variance $1/n_r$, then the empirical distribution of the eigenvalues of 
$\mh^H\mh$ converges almost surely, as 
$n_r,n_t \rightarrow \infty$ with $n_t/n_r = \alpha$, to 
a density function \cite{chard3}
\begin{equation}
p_{\alpha}(x) = \Big(1-\frac{1}{\alpha}\Big)^+ \delta(x) + \frac{\sqrt{(x-a)^+(b-x)^+}}{2\pi\alpha x},
\label{marc}
\end{equation}
where $(x)^+= \max(x,0)$, $a=(1-\sqrt{\alpha})^2$, and $b=(1+\sqrt{\alpha})^2.$
An effect of the Mar\v{c}enko-Pastur law is that very tall or very wide 
matrices\footnote{In practice, the channel matrix in a multiuser system with 
tens of single-antenna users and hundreds of BS antennas will become a very 
tall matrix on the uplink, and a very wide matrix on the downlink.} are very 
well conditioned. The law also implies that the channel ``hardens'', i.e., 
the eigenvalue histogram of a single realization converges to the average 
asymptotic eigenvalue distribution.  

Channel hardening can bring in several advantages in large dimensional 
signal processing. For example, linear detection in large systems will 
require inversion of large matrices. Inversion of large random matrices 
can be done fast using series expansion techniques 
\cite{invert1},\cite{invert2},\cite{invert3}. 
Because of channel hardening, approximate matrix inversions using series 
expansion and deterministic approximations from limiting distribution 
become effective in large dimensions.

An interesting aspect in channel hardening is that as the size of $\mh$ 
increases, the off-diagonal terms of the $\mh^H\mh$ matrix become increasingly 
weaker compared to the diagonal terms, i.e., 
$\frac{\mh^H\mh}{n_r}\rightarrow\mi_{n_t}$  for $n_r,n_t \rightarrow \infty$ 
with $n_t/n_r=\alpha$. This phenomenon is pictorially illustrated in Fig. 
\ref{hardening}, where we have plotted $\mh^T\mh$ for the 
real-valued channel model in (\ref{sys}) for $8\times 8 $, $32\times 32$,
$64\times 64$, and $128\times 128$ channels. In proposing the new receiver
algorithm in the next section, we will work with approximations to the 
off-diagonal terms of the $\mh^T\mh$ matrix and estimates of $\mh^T\mh$, 
which are found to achieve very good performance in large dimensions at 
low complexities.

\begin{figure}
\hspace{-1mm}
\includegraphics[width=3.65in,height=3.5in]{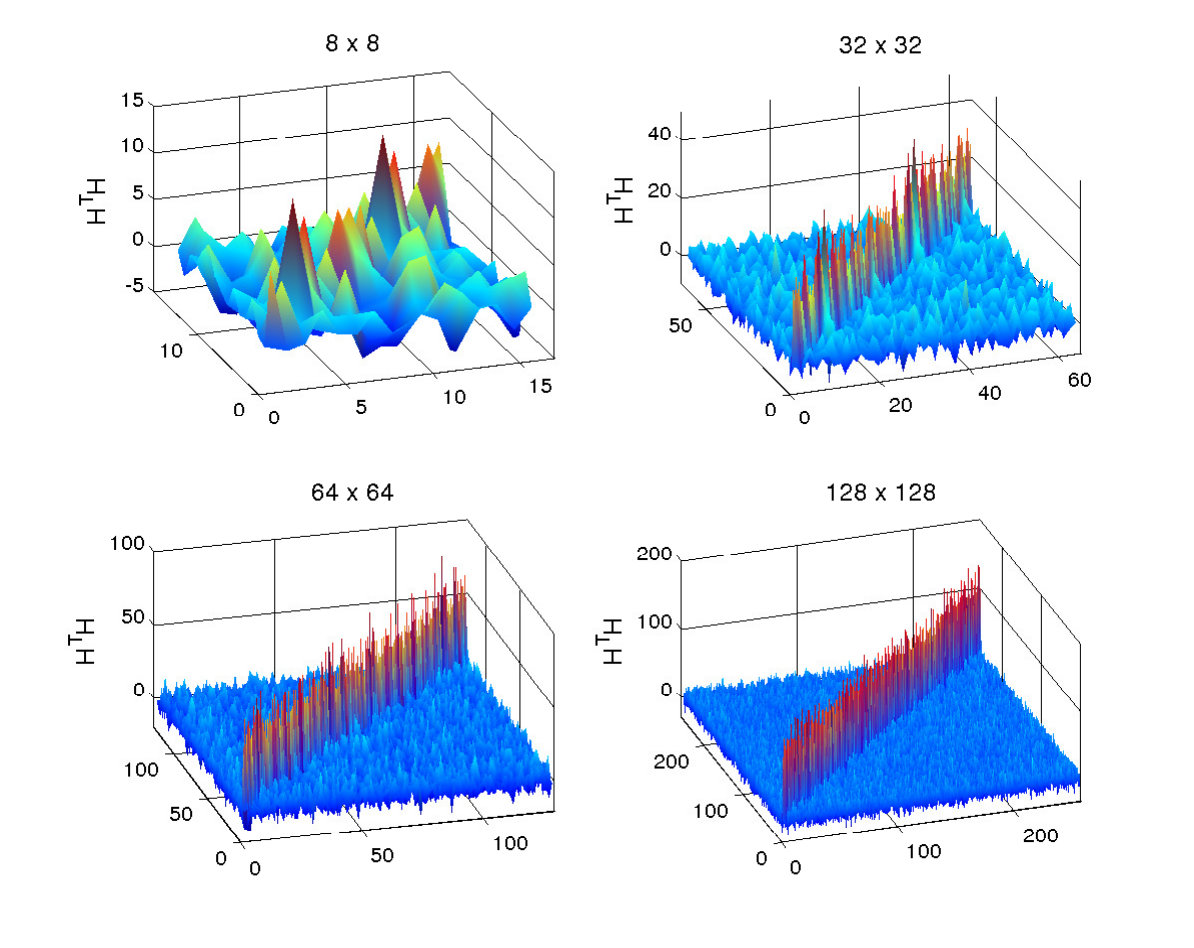}
\vspace{-6mm}
\caption{ \label{hardening}
Magnitude plots of ${\bf H}^T{\bf H}$ for $8\times 8$, $32\times 32$, 
$64\times 64$, and $128\times 128$ MIMO channels.}
\end{figure}

\section{The Proposed CHEMP Receiver }
\label{sec3} 
In this section, we present the proposed CHEMP receiver. The 
proposed CHEMP receiver has two main components: $1)$ a message passing 
based detection (MPD) algorithm, and $2)$ a scheme to estimate $\mh^T\mh$. 
The proposed MPD algorithm works with the real-valued matched filtered 
received vector (whose signal term becomes ${\bf H}^T{\bf H}{\bf x}$), 
and uses a Gaussian approximation on the off-diagonal terms of the  
${\bf H}^T{\bf H}$ matrix.

Before we describe the proposed MPD algorithm, we state the following lemma 
which will be used in the development and analysis of the detection algorithm.
\begin{lem}
\label{lemma1}
Let $X_i$ and $Y_i$ be Gaussian random variables with zero mean and variance 
$\sigma^2_x$ and $\sigma^2_y$, respectively. Let $Z_i\Define X_iY_i$ and
$Z\Define\frac{1}{n}\sum\limits_{i=1}^nZ_i$. 
\begin{itemize}
\item 	When $X_i$ and $Y_i$ are independent, $\E Z_i=0$ and 
	$\E Z_i^2=\sigma^2_x\sigma^2_y$. Then by central limit theorem,
	for large $n$, $Z\sim\mathcal{N}(0,\frac{\sigma^2_x\sigma^2_y}{n})$. 
	When $X_i$ and $Y_i$ are i.i.d., 
	$Z\sim\mathcal{N}(0,\frac{\sigma^4_x}{n})$.
\item 	When $X_i=Y_i$, $Z$ is a $\chi^2$ random variable of degree $n$. 
	$\E Z=\sigma^2_x$ and Var$(Z)=\frac{2\sigma^4_x}{n}$. \qedsymbol
\end{itemize} 
\end{lem}

\subsection{Proposed MPD algorithm}
\label{sec3a}
Consider the real-valued system model in (\ref{sys}).
 We consider 
4-QAM modulation in this section, i.e., 
${\mathbb B}=\{\pm 1 \}$. We will extend the algorithm to higher-order
QAM in Section \ref{subsec_new}.
Performing matched filter operation on ${\bf y}$, we have
\begin{equation}
\mh^T\vy = \mh^T\mh\vx + \mh^T\vw.
\label{mfeq}
\end{equation}
From (\ref{mfeq}), we write the following:
\begin{equation}
\vz =\mj\vx+\vv,
\label{eqn1}
\end{equation}
where
\begin{equation}
\vz \Define \frac{\mh^T\vy}{N}, \quad
\mj\Define\frac{\mh^T\mh}{N}, \quad \vv\Define\frac{\mh^T\vw}{N}.
\label{def1}
\end{equation}
The $i$th element of $\vz$ can be written as
\begin{equation}
z_i=J_{ii}x_i+\underbrace{\sum_{j=1, j\neq i}^{2K} J_{ij}x_j + v_i}_{\Define \ g_i},
\label{eqn2}
\end{equation}
where $J_{ij}$ is the element in the $i$th row and $j$th column of 
$\mj$, $x_i$ is the $i$th element of $\vx$, and 
\begin{equation}
v_i=\sum_{j=1}^{2N}\frac{H_{ji}w_j}{N}
\end{equation}
is the $i$th element of $\vv$, where $H_{ji}$ is the $(j,i)$th element of
${\bf H}$. Note that the variable $g_i$ defined in
(\ref{eqn2}) denotes the interference-plus-noise term, which involves the 
off-diagonal elements of $\frac{\mh^T\mh}{N}$ (i.e., $J_{ij}$, $i\neq j$).
We approximate the $g_i$ term to have a a Gaussian distribution with mean 
$\mu_i$ and variance $\sigma^2_i$, i.e., the distribution of $g_i$ is 
approximated as $\mathcal{N}(\mu_i,\sigma^2_i)$. By central limit theorem, 
this approximation is accurate for large $K$, $N$. The mean and variance 
in this approximation are given by
\begin{eqnarray}
\mu_i&=&\E(g_i)=\sum_{j=1, j\neq i}^{2K} J_{ij} \E(x_j) \label{eq_mu} \\
\sigma^2_i&=&\text{Var}(g_i)=\sum_{j=1, j\neq i}^{2K} J_{ij}^2 \text{Var}(x_j) + \sigma^2_{v}.
\label{musigma}
\end{eqnarray}
Denoting the probability of the symbol $x_j$ as $p_j$, we have
\begin{eqnarray}
\E(x_j) = (2p_j-1), \quad \text{Var}(x_j)= 4 p_j(1-p_j). 
\end{eqnarray}
Also, note that by Lemma \ref{lemma1}, $\sigma^2_{v}=\frac{\sigma^2_n}{2N}$.
Because of the above Gaussian approximation, the a posteriori probability 
(APP) of the symbol $x_i$ can be written as
\begin{eqnarray}
\hspace{-4mm}
p_i\ = \ \Pr(x_i|z_i, \mj) & \hspace{-2mm} \propto& \hspace{-2mm} \exp\Big(\frac{-1}{2\sigma^2_i}(z_i-J_{ii}x_i-\mu_i)^2\Big).
\label{probs}
\end{eqnarray}
From (\ref{probs}), the log-likelihood ratio (LLR) of $x_i$, denoted by
$L_i$, can be written as
\begin{eqnarray}
L_i & = & \ln\frac{\Pr(z_i|x_i=+1)}{\Pr(z_i|x_i=-1)} \nonumber \\
    & = &\frac{2J_{ii}}{\sigma^2_i}(z_i-\mu_i). 
\label{llreq}
\end{eqnarray}
From (\ref{llreq}), 
the probability of symbol $x_i$, can be written as
\begin{eqnarray}
p_i&=&\frac{e^{L_i}}{1+e^{L_i}}.
\label{pix}
\end{eqnarray}

\begin{figure}
\includegraphics[width=3.25in,height=1.90in]{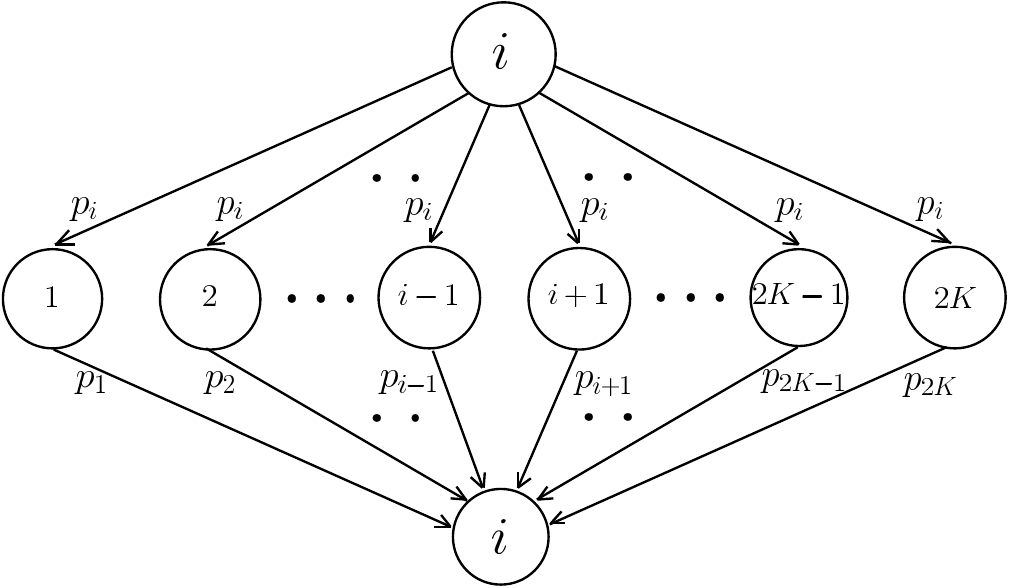} 
\caption{Message passing in the proposed MPD algorithm.}
\label{msgpass} 
\end{figure}

{\em Message passing:}
The system is modeled as a fully-connected graph, where the data symbols in 
$\vx$ represent the nodes. 
There are $2K$ nodes in the graph corresponding to the $2K$ elements in 
the vector $\vx$. The $i$th node uses the knowledge of $\mj$, $\vz$ and 
the incoming APPs $\{p_1,p_2,\cdots,p_{i-1},p_{i+1},\cdots,p_{2K}\}$ 
to obtain a soft estimate of the interference to symbol $x_i$, and 
computes its APP, $p_i$. That is, each node is an approximate APP 
processor for its associated symbol, and message passing refers to the 
exchange of APP values computed at each iteration. 
Figure \ref{msgpass} illustrates the above message passing schedule. 
Note that the computation of the message $p_i$ in (\ref{pix}) requires 
the computation of (\ref{eq_mu}), (\ref{musigma}) and (\ref{llreq}). 
The algorithm is 
initialized with $p_i=0.5$, $\forall i$, and message passing is carried 
out for a certain number of iterations, after which the algorithm stops. 
The values of $p_i$s at the end are taken 
as the soft values of $x_i$s. These soft values can be directly fed to the 
channel decoder in coded systems. In uncoded systems, a hard estimate of 
symbol $x_i$ can be obtained as
\begin{equation}
\hat{x}_i=\left\{ 
\begin{array}{ll}
+1 & \mbox{if } p_i \geq 0.5\\
-1 & \mbox{otherwise.}
\end{array} \right.
\label{bits}
\end{equation}

\subsection{Improving convergence rate}
\label{sec3b}
At the end of the $t$th iteration of the detection algorithm described 
above, we obtain the probability of the $i$th user's information bit, 
$p^t_i$. The rate of convergence of this sequence 
$\{p^0_i,p^1_i,p^2_i,\cdots,p^t_i,\cdots\}$ 
can be improved by certain techniques. We discuss the following two 
techniques that helps us to improve the convergence. 
\begin{itemize}
\item
{\em Aitken acceleration}: Aitken's delta-squared process is a technique 
known in numerical analysis \cite{aitken} for accelerating sequence 
convergence. This method is also used in \cite{bickson} to accelerate 
the convergence of the Gaussian belief propagation algorithm. By this 
method, a linearly converging sequence of real numbers can be accelerated 
to converge quadratically. Although there is no rigorous proof guaranteeing 
this rate of convergence, empirical observations have shown that this method 
does accelerate the convergence of iterative algorithms. According to 
Aitken's acceleration method, we define a sequence
\begin{equation}
q_i^t=p_i^{t}-\frac{(p_i^{t+1}-p_i^{t})^2}{p_i^{t+2}-2p_i^{t+1}+p_i^{t}}.
\label{aitken}
\end{equation}
This new sequence $q_i^{t}$ converges faster than $p_i^{t}$ and to the same 
limit, whenever $p_i^{t}$ converges. After the first three iterations, 
$q_i$s can be used as the messages in the algorithm for faster convergence.

\item
{\em Damping}: Damping of messages passed in message passing algorithms is a 
scheme known to improve the rate of convergence of iterative algorithms 
\cite{damping}. At the $t$th iteration, the message is damped by obtaining a 
convex combination of the message computed at the $t$th iteration and the 
message at the $(t-1)$th iteration, with a damping factor $\Delta\in[0,1)$. 
Thus, if $\tilde{p}_i^t$ is the computed probability at the $t$th iteration, 
the message at the end of $t$th iteration is
\begin{equation}
p_i^{t}=(1-\Delta)\tilde{p}_i^t+\Delta p_i^{t-1}.
\label{damp}
\end{equation}
\end{itemize}
In section \ref{sec3d}, we will see the performance of these methods in 
improving the rate of convergence and the optimal choice for $\Delta$.

A listing of the proposed MPD algorithm with damping is given in 
{\bf Algorithm \ref{mpmdalgo}}, where $\vp=[p_1 \ p_2 \ \cdots \ p_{2K}]^T$
and $\tilde\vp=[\tilde p_1 \ \tilde p_2 \ \cdots \ \tilde p_{2K}]^T$.

\begin{algorithm}[t]
\caption{Proposed MPD algorithm}
\label{mpmdalgo}
\begin{algorithmic}[1]
\REQUIRE{$\vz$ , $\mj$, $\sigma^2_v$, $\Delta$}
\STATE {\bf Initialize}: $p_i^0\gets 0.5$, $i=1,\cdots,2K$
\FOR {$t = 1$ \TO  {\em number\_of\_iterations}}
{
\FOR{$i = 1$ \TO $2K$}{
\vspace{1mm}
\STATE $\mu_{i} \gets \sum\limits_{j=1,j\neq i}^{2K} J_{ij}(2p_j^{t-1}-1)$
\vspace{1mm}
\STATE $\sigma_{i}^2 \gets \sum\limits_{j=1,j\neq i}^{2K} 4J_{ij}^{2}p_j^{t-1}(1-p_j^{t-1})+\sigma^{2}_v$  
\vspace{1mm}
\STATE $L_i \gets \frac{2J_{ii}}{\sigma_i^2}(z_i-\mu_i)$
\vspace{1mm}
\STATE $\tilde{p}_i^t \gets \frac{e^{L_i}}{1+e^{L_i}}$ 
}
\ENDFOR
\STATE $\vp^{t} \gets (1-\Delta)\tilde{\vp}^t+\Delta \vp^{t-1}$
}
\ENDFOR
\end{algorithmic}
\end{algorithm}

\subsection{Complexity comparison between MPD and MMSE}
\label{sec3c}
The computational complexity of the MPD algorithm is as follows. The 
complexity (in number of real operations) required to compute 
(\ref{eq_mu}), (\ref{musigma}) 
and (\ref{pix}) is of order $O(K^2)$. The complexities of computing $\vz$ 
and $\mj$ are of orders $O(NK)$ and $O(NK^2)$, respectively. So, the total 
complexity of the proposed MPD is $O(NK^2)$, which is attractive for 
large-scale MIMO systems.

In Table \ref{comp1}, we present an interesting comparison between the
complexities of MPD and MMSE detection for $N=128,256$, and $K$
varied from 16 to 256. Since we have used 20 iterations for MPD in
all the BER simulations, we have taken the number of iterations to be 20
for the calculation of the MPD complexity. From Table \ref{comp1}, the
following interesting observations can be made: 1) for large $N$ (e.g.,
$N=256$), MPD complexity is less than MMSE complexity. This is because
MPD needs only matrix multiplication and not matrix inversion, whereas
MMSE detection needs both matrix multiplication and inversion; 2) for $N=128$, 
the MPD complexity for $K=64,96,128$ is less than the MMSE complexity. For 
$K=16,32$, the MPD complexity is almost the same as (marginally higher than) 
MMSE complexity, because the number of iterations ($=20$) is comparable 
with $K$ ($=16,32$). Also, MPD performs better than MMSE detection, and 
achieves close to optimal detection performance for large $K,N$, and different 
system loading factors. We will see this performance advantage of MPD in the 
following subsection. 

\begin{table}
\centering
\begin{tabular}{|c||c|c||c||c|c|c|}
\hline
&\multicolumn{6}{|c|}{Complexity in number of real operations $\times 10^6$}\\
\cline{2-7} $K$&\multicolumn{3}{|c|}{$N=128$}&\multicolumn{3}{|c|}{$N=256$}\\
\cline{2-7} & MMSE & MPD & SUMIS & MMSE & MPD & SUMIS\\
 & & (prop) &  in \cite{sumis} & & (prop)& in \cite{sumis} \\
\hline\hline
16&        0.177 & 0.179 & 0.483 & 0.333 & 0.296 & 0.917 \\ \hline
32&        0.748 & 0.749 & 1.737 & 1.321 & 1.190 & 3.130 \\ \hline
64&        3.593 & 3.200 & 7.538 & 5.789 & 4.773 & 12.420 \\ \hline
96&        9.584 & 7.208 & 19.368 & 14.450 & 10.748 & 29.837 \\ \hline
128&       19.770& 12.814& 39.194 & 28.355& 19.116 & 57.347 \\ \hline
256& -& -&  - & 157.373 & 76.505 & 307.633 \\ \hline
\end{tabular}
\caption{Comparison between the complexities (in number of real operations)
of the proposed MPD, MMSE detector, and SUMIS detector in \cite{sumis} for 
different values of $K,N$. Number of iterations for MPD = 20, and $n_s=3$ for 
SUMIS.}
\label{comp1}
\end{table}

\subsection{BER performance of MPD}
\label{sec3d}
In this subsection,
we present the uncoded BER performance of MPD obtained through simulations
for different system parameter settings.
We will now assume perfect knowledge ${\bf H}$. We will relax this assumption 
later.  First, in Fig. \ref{simmpmd_damp}, we plot the uncoded BER of MPD at 
an average SNR of 12 dB for $N=K=64$ for various values of the damping factor 
$\Delta$. The number of message passing iterations used is 20. From this 
figure, we observe that a damping factor of $\Delta=0.33$ is optimal. This 
value of $\Delta$ is found to give good performance for other values of 
system parameters as well. So we have used this value of $\Delta$ in all the 
simulations. Next, Fig. \ref{simmpmd_aitken} shows the uncoded BER of MPD as 
a function of iteration index with and without Aitken acceleration for 
$N=K=64$, SNR=12 dB, and $\Delta=0.33$. It can be observed that the 
convergence rate of the algorithm improves with Aitken acceleration. 

\begin{figure}
\includegraphics[width=3.6in,height=2.6in]{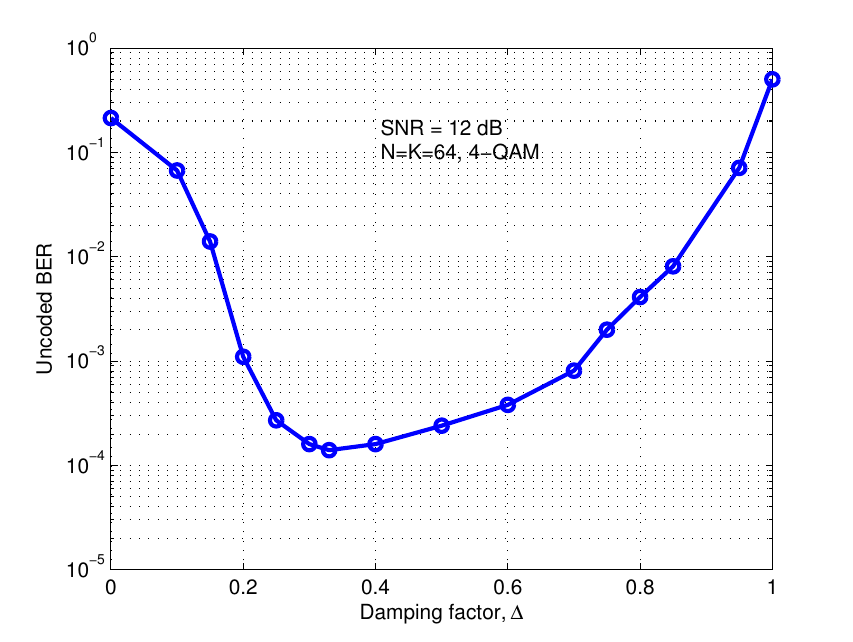} 
\vspace{-4mm}
\caption{
Uncoded BER performance of the proposed MPD algorithm as a 
function of damping factor $\Delta$. $N=K=64$, 4-QAM, SNR=12 dB.
}
\label{simmpmd_damp} 
\end{figure}

\begin{figure}
\includegraphics[width=3.6in,height=2.6in]{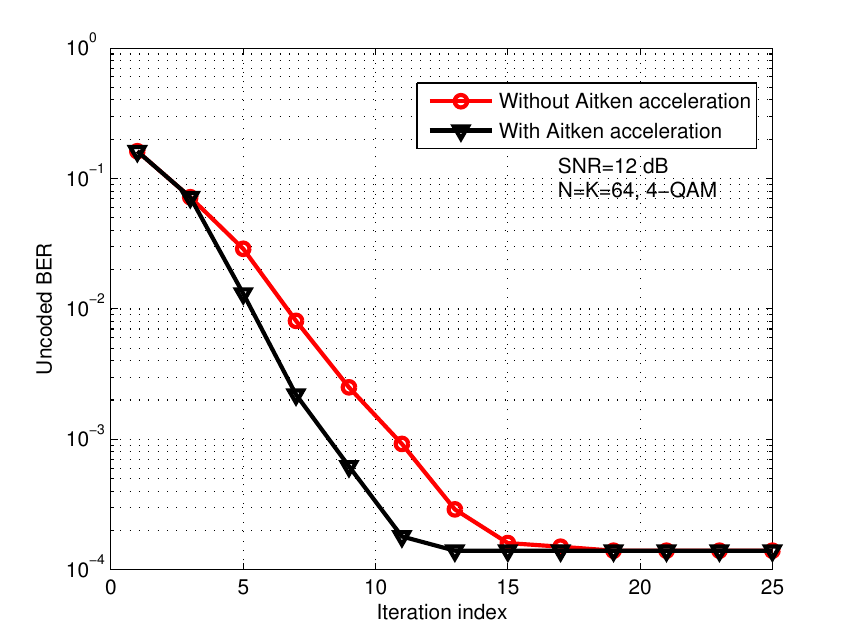} 
\vspace{-4mm}
\caption{Comparison of the convergence behavior of the MPD algorithm without 
and with Aitken acceleration. $N=K=64$, 4-QAM, SNR=12 dB.}
\label{simmpmd_aitken} 
\end{figure}

In Fig. \ref{simmpmd_N}, we plot the uncoded BER of MPD for different 
values of $N$ ($=4,8,16,32,64,128$) for a system loading factor of 
$\alpha=1$ ($K=N$). Since optimal detection performance for large-dimension 
systems is hard to obtain, we have plotted single-input single-output (SISO) 
additive white Gaussian noise (AWGN) channel performance as a lower bound on 
the optimum detection performance. MMSE detection performance 
is also plotted for comparison. From Fig. \ref{simmpmd_N}, it is observed 
that the performance of MPD improves for increasing $N,K$, and moves closer 
to the SISO-AWGN performance for large $N,K$. For example, the MPD performance 
for $N=K=128$ gets very close to SISO-AWGN performance. It is 
also observed that MPD performance is better than MMSE detection performance.

\begin{figure}
\includegraphics[width=3.6in,height=2.6in]{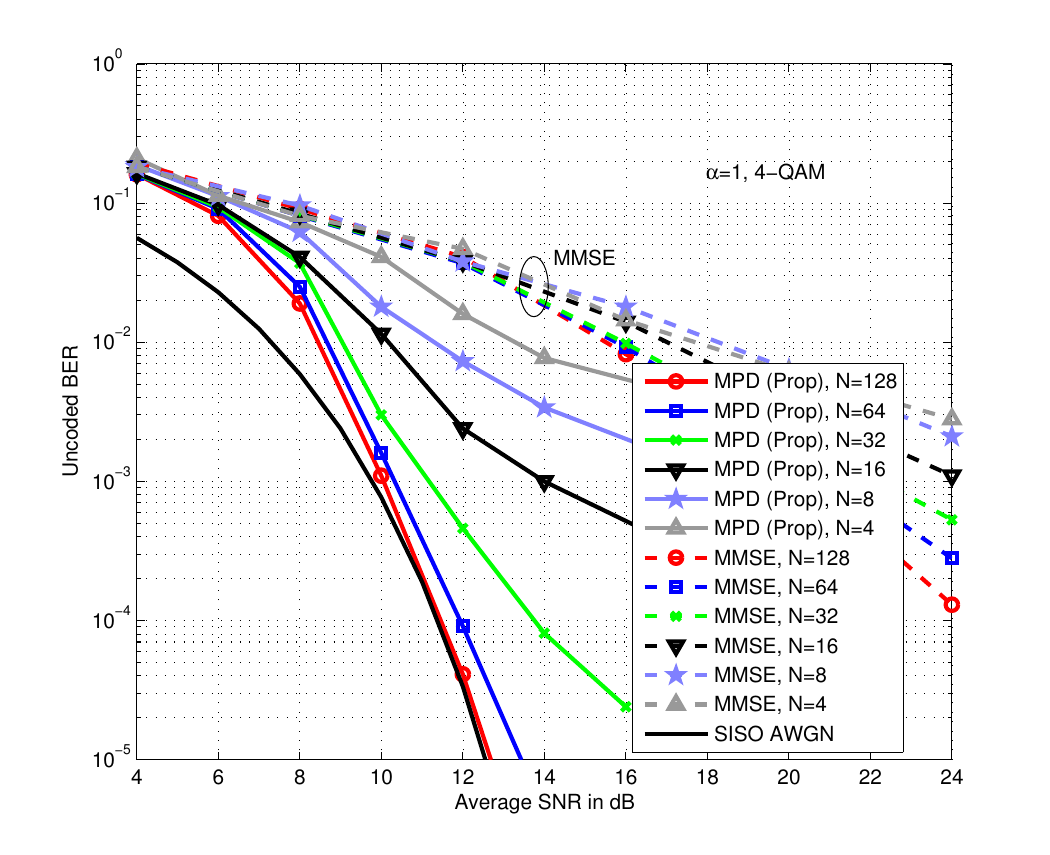} 
\vspace{-4mm}
\caption{
Uncoded BER performance of the MPD algorithm and the MMSE detector
for $N=K=4,8,16,32,64,128$, 4-QAM.
}
\label{simmpmd_N} 
\end{figure}

\begin{figure}
\centering
\includegraphics[width=3.6in,height=2.6in]{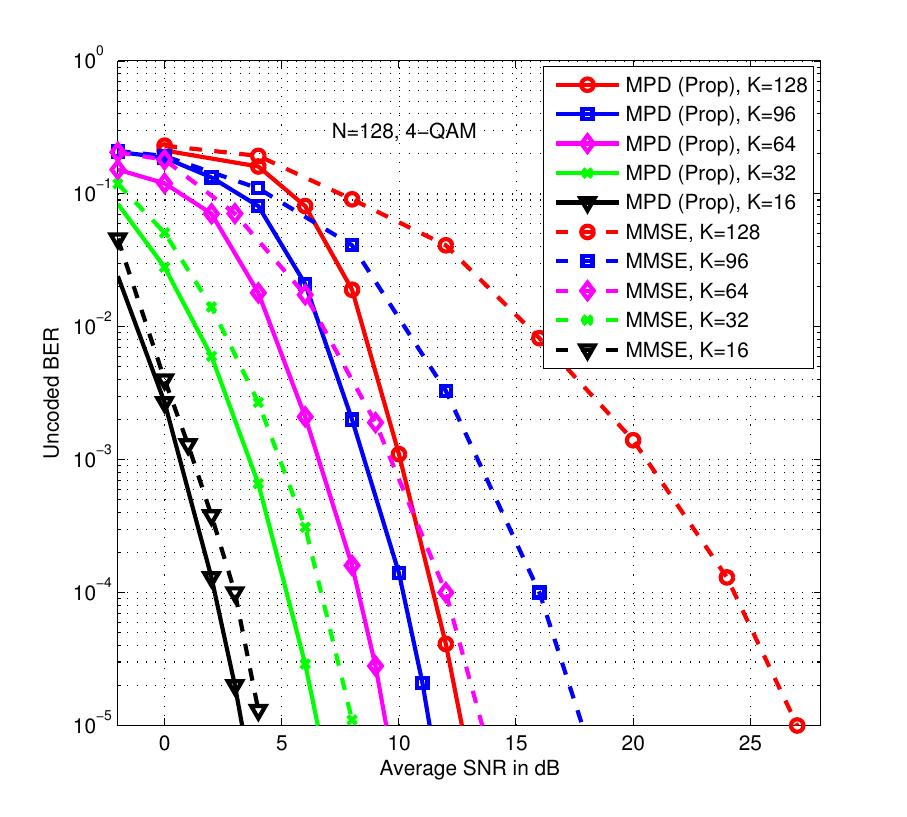}
\vspace{-4mm}
\caption{
Uncoded BER performance of the MPD algorithm and the MMSE 
detector for different values of $K\ (=16,32,64,96,128)$ for a fixed 
$N= 128$, 4-QAM.}
\label{simmpmd_K}
\end{figure}

\begin{figure}[h]
\includegraphics[width=3.6in,height=2.6in]{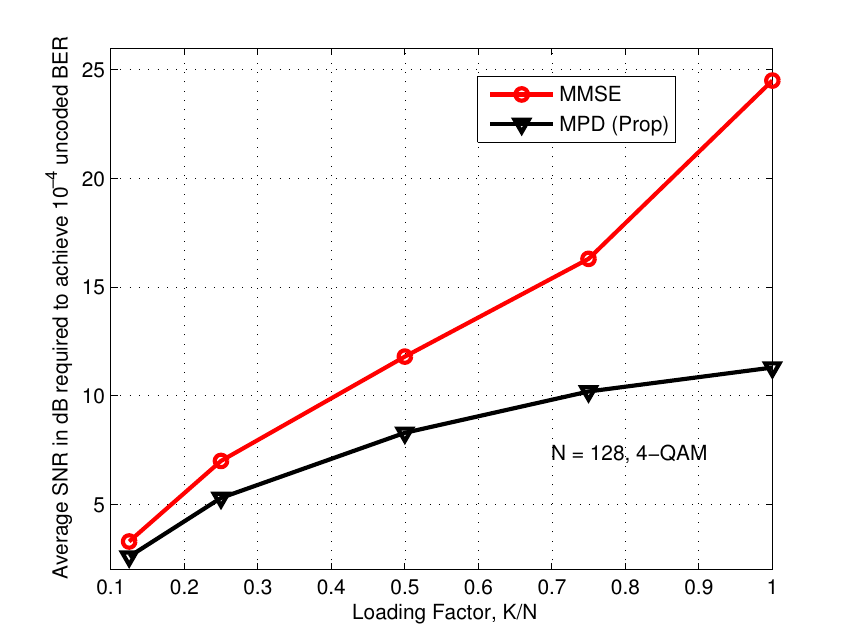} 
\caption{
Comparison between the average SNR required to achieve an uncoded 
BER of $10^{-4}$ in MPD and MMSE detection at different loading factors 
with $N=128$, 4-QAM.
}
\label{simmpmd_lf} 
\end{figure}

Figure \ref{simmpmd_K} shows the uncoded BER of MPD algorithm 
and MMSE detector for a fixed number of receiver antennas 
at the BS ($N=128$) and varying number of users ($K=16,32,64,96,128$), 
i.e., for different values of loading factors 
($\alpha=\frac{1}{8},\frac{1}{4},\frac{1}{2},\frac{3}{4},1$). It is 
observed that the BER performance improves considerably as the loading 
factor is reduced, which is expected. 
The MPD performance for different loading factors is 
better than MMSE detection performance. This observation is further 
illustrated in Fig. \ref{simmpmd_lf}, where the average SNRs required to 
achieve an uncoded BER of $10^{-4}$ in MPD and MMSE detection are plotted. 
It can be observed from Fig. \ref{simmpmd_lf} 
that the MPD outperforms the MMSE detection by about 1.2 dB at a loading 
factor of $\alpha=0.125$. This performance advantage of MPD over MMSE 
detection increases for increasing values of $\alpha$. For example, the 
performance advantage of MPD over MMSE detection is about 6.5 dB and 
12.5 dB for $\alpha=0.75$ and $\alpha=1$, 
respectively. The reason why MMSE detection performs quite poorly 
at high loading factors is because the spatial interference gets increased 
significantly at higher loading factors with large $N$ (e.g., $N=K=128$)
compared to lower loading factors, and MMSE detection does not 
perform interference cancellation/suppression. Whereas, the MPD is benefited
by the channel hardening effect with large $N,K$. The performance 
advantage of MPD becomes very attractive given that MPD complexity is almost 
same or less than the MMSE detection complexity (as discussed in Section 
\ref{sec3c}).

\begin{figure}
\includegraphics[width=3.6in,height=2.6in]{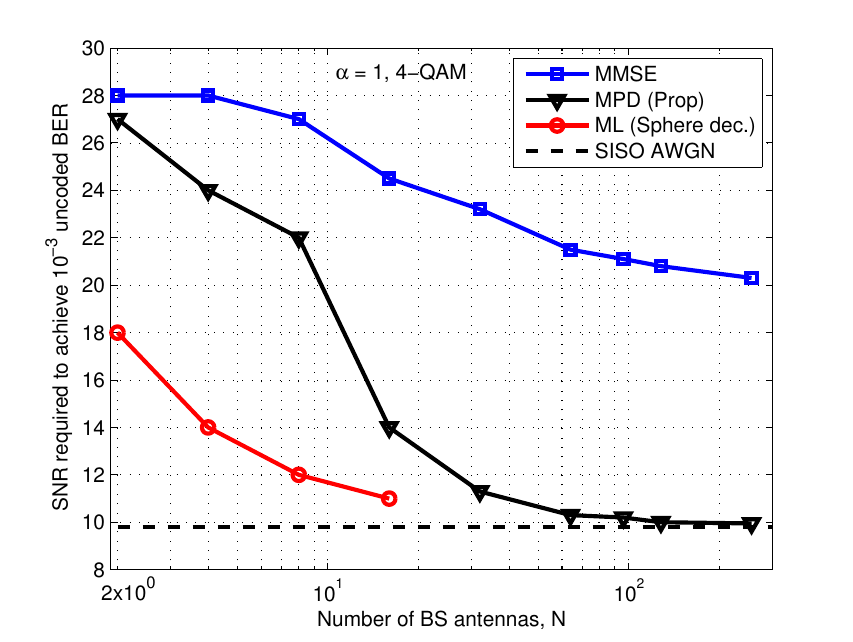} 
\vspace{-4mm}
\caption{
Comparison between the average SNR required to achieve an uncoded 
BER of $10^{-3}$ BER in ML (sphere decoding), MPD, and MMSE detection 
as a function of $N$ for $\alpha=1$ (i.e., $N=K$) and 4-QAM.
}
\vspace{-8mm}
\label{fig_sd} 
\end{figure}

The effect of channel hardening on the BER performance of the MPD algorithm
is further illustrated in Fig. \ref{fig_sd}. This figure shows the SNRs 
required to achieve $10^{-3}$ BER with MPD as well as MMSE detection in
$N=K=2$ to $N=K=256$ systems. We have also plotted the same for ML detection 
(using sphere decoding) in $N=K=2$ to $N=K=16$ systems. Since ML detection 
is prohibitive for larger dimensions, we have plotted the SNR required in
a SISO AWGN system as a lower bound on the ML performance. In small systems 
like $N=K=2,4,8$ systems where channel hardening is not significant, both 
MPD and MMSE performances are far from ML performance with MPD performing 
better than MMSE -- e.g., MPD performance is about 10 dB away from ML 
performance in $N=K=4,8$ systems, whereas MMSE performance is about 14 to 
15 dB away from ML performance in $N=K=4,8$ systems. In systems with size 
larger than $N=K=16$, channel hardening becomes more significant and the 
performance of MPD shows significant improvement compared to MMSE and gets 
closer to ML performance -- e.g., for $N=K=128$ system, the MPD performance 
is just about 0.25 dB away from the ML lower bound whereas the MMSE 
performance is away from the ML lower bound by about 10 dB. These 
observations illustrate that harder the channel gets, better is the MPD 
performance.

\subsection{Channel estimation for MPD}
\label{sec3e}
A key issue in large-scale MIMO systems is the estimation of channel gains. 
In conventional approaches, the $NK$ channel gains in the channel matrix are 
estimated and used for the detection of transmitted symbols. Note that in our 
transformed system model (\ref{eqn1}), the influence of the channel on vector
$\vz$ is through $\mh^T\mh$, rather than through $\mh$ as such. We propose 
to exploit this observation on the structure of the system model (\ref{eqn1}). 
Specifically, we propose to directly obtain an estimate of $\mh^T\mh$ and
use it in the MPD algorithm, rather than obtaining an estimate of $\mh$ as
done in conventional approaches. We note that this approach is simple and 
novel, and it works very well in the MPD algorithm (as we will see in the
performance results). We present the scheme to obtain an estimate of the
$\mh^T\mh$ matrix next.

{\em Estimating the $\mh^T\mh$ matrix:}

Note that we have defined $\mj=\mh^T\mh$. We are interested in obtaining
$\hat{\mj}$, an estimate of $\mj$. 
We assume that the channel is slowly fading, where the channel matrix 
$\mh$ remains constant over one frame duration (which is taken to be equal
to the coherence time of the channel). The length of one frame is $L_f$ 
channel uses. Each frame consists of a pilot part and a data part. The 
pilot part consists of $K$ channel uses, and the data part consists of 
$L_f-K$ channel uses. 

Let $\mx_p=P\mi_{K}$ denote the pilot matrix, where in the $i$th channel 
use, $1\leq i\leq K$, user $i$ transmits a pilot tone with amplitude $P$ 
and the other users remain silent. The received pilot matrix at the BS is 
then given by
\begin{eqnarray} 
\my_p & = & \mh\mx_p+\mw_p \nonumber \\
& = & P\mh+\mw_p,
\end{eqnarray}
where $P=\sqrt{KE_s}$, $E_s$ is the average symbol energy, and $\mw_p$ is 
the noise matrix. Using Lemma \ref{lemma1}, we obtain an estimate of the 
matrix $\mj$ as 
\begin{eqnarray}
\mje & = & \frac{\my_p^T\my_p}{NP^2}-\frac{\sigma^2_v}{P^2}\mi_K.
\label{mpmdj}
\end{eqnarray}
An estimate of the vector $\vz$ is obtained as
\begin{eqnarray}
\vze & =  & \frac{\my_p^T\vy}{NP}.
\label{mpmdz}
\end{eqnarray}
The estimates $\mje$ and $\vze$ are used as inputs to the MPD algorithm
in place of $\mj$ and $\vz$.

{\em Note on complexity:}

A key advantage of the above estimation scheme is its low complexity.
The computation of $\mje$ and $\vze$ in (\ref{mpmdj}) and (\ref{mpmdz}) 
requires only matrix and vector multiplications. Note that even when 
perfect knowledge of $\mh$ or an estimate of $\mh$ is available, similar 
computations are needed to compute $\mj$ and $\vz$. Further note that 
the additional complexity needed to obtain an estimate of $\mh$ in the
conventional approach is avoided in our approach.

\subsection{BER performance of the CHEMP receiver}
\label{sec3f}
As mentioned before, we refer to the combination of proposed MPD algorithm 
and the channel estimation scheme proposed in the previous subsection as 
the CHEMP receiver. In this subsection, we present the uncoded BER performance 
of the CHEMP receiver. The number of iterations used in the MPD algorithm is
20. We compare the performance of the CHEMP receiver with two other receivers,
namely, 1) MMSE detector with MMSE channel estimate, and 2) FG-GAI (factor 
graph with Gaussian approximation of interference) detector in \cite{jstsp} 
with MMSE channel estimate. We note that the FG-GAI detector in \cite{jstsp}
is also a message passing algorithm which used a Gaussian approximation of 
interference. But this approximation was done on the original system model 
in (\ref{sys}), whereas in the proposed MPD algorithm, the Gaussian 
approximation is done on the matched filtered system model in (\ref{eqn1}).

In Fig. \ref{simmpmdce}, we present an uncoded BER performance comparison
between 1) proposed CHEMP receiver, 2) MMSE detector with MMSE channel
estimate, and 3) FG-GAI detector in \cite{jstsp} with MMSE channel estimate.
It can be seen that the performance of the proposed  CHEMP receiver is 
significantly better than those of the MMSE and FG-GAI detectors with 
MMSE estimate of the channel. Observe that the performances 
of MPD and FG-GAI under perfect CSI conditions are almost the same, whereas 
under estimated CSI conditions, the CHEMP receiver performs significantly 
better than FG-GAI with MMSE channel estimate. An analytical reasoning for 
this is presented in Section \ref{sec4b}. 

\begin{figure}
\includegraphics[width=3.6in,height=2.6in]{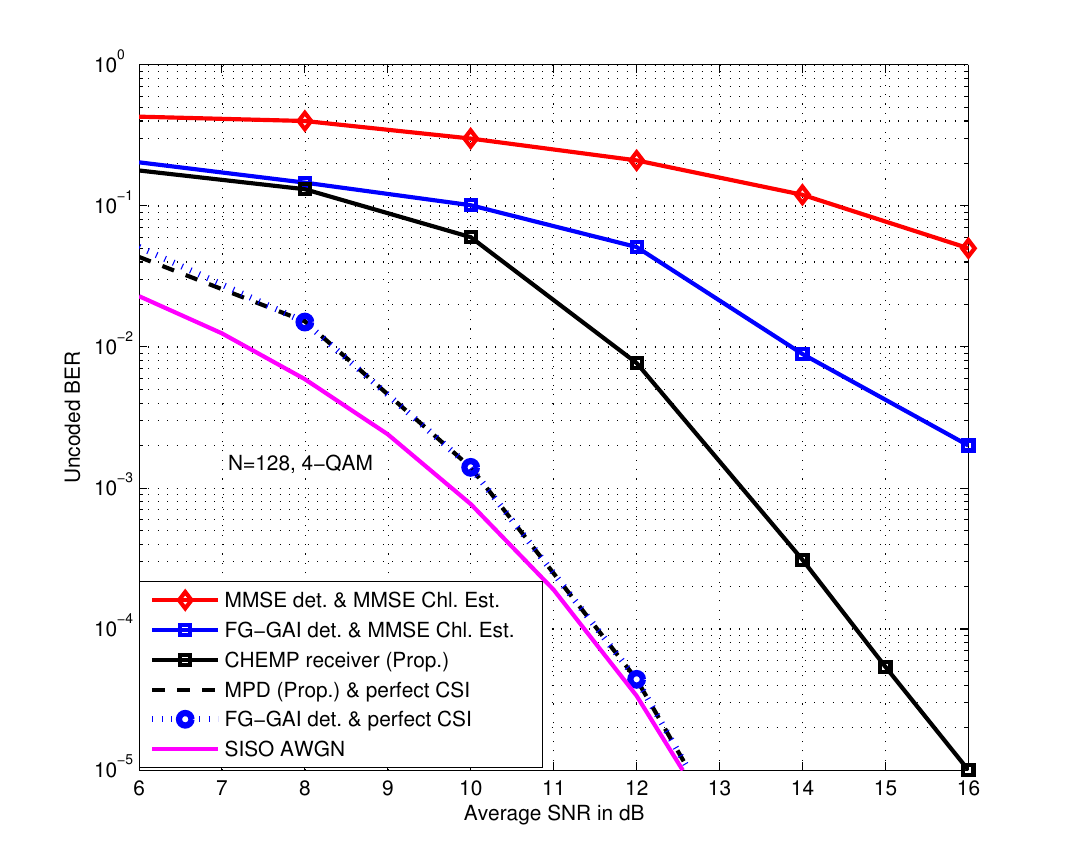} 
\vspace{-4mm}
\caption{
Comparison of the BER performance of the proposed CHEMP receiver
with those of 1) MMSE detector with MMSE channel estimate, and 2) FG-GAI
detector in \cite{jstsp} with MMSE channel estimate, for $N=K=128$, 4-QAM.
}
\label{simmpmdce} 
\end{figure}

\begin{figure}
\includegraphics[width=3.6in,height=2.6in]{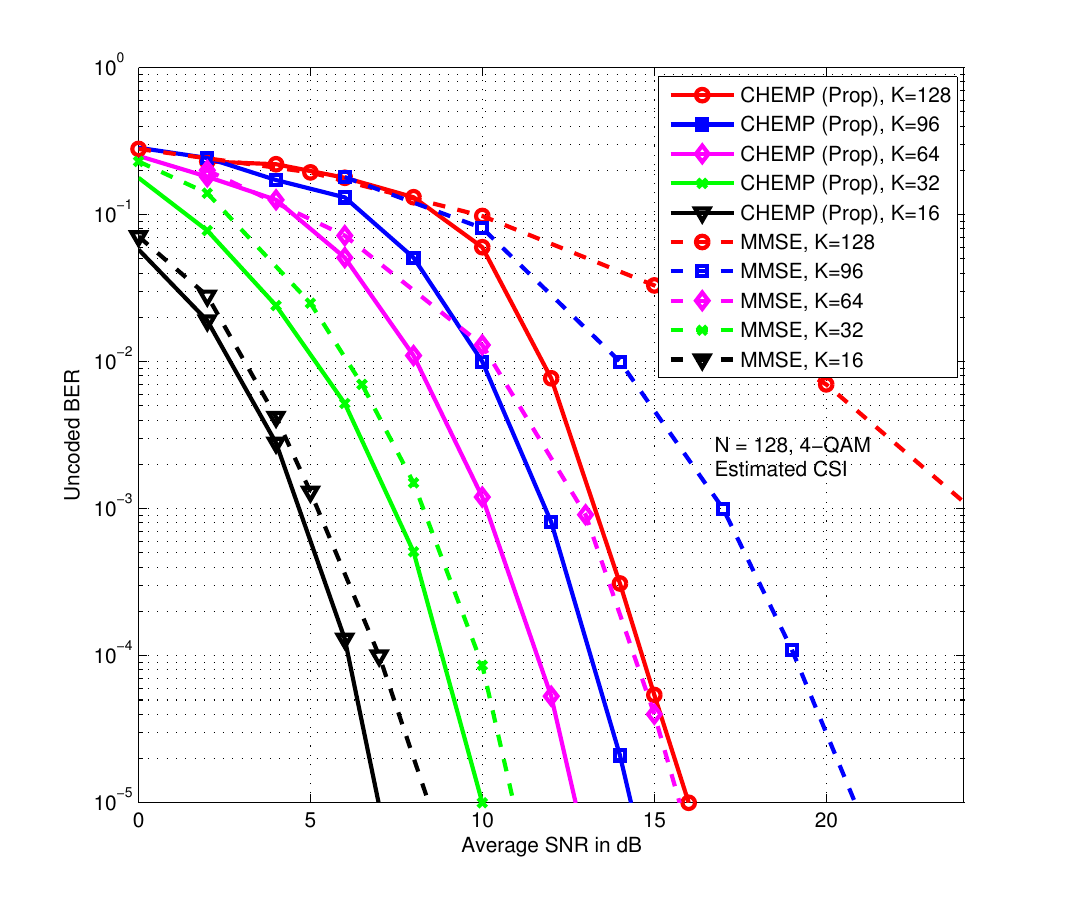} 
\vspace{-4mm}
\caption{
BER performance of 1) proposed CHEMP receiver and 2) MMSE detector
with MMSE channel estimate, for different values of $K$ \ 
($=16,32,64,96,128$) for a fixed value of $N$ \ ($=128$), 4-QAM.
}
\label{simmpmdce_K} 
\end{figure}

\begin{figure}
\includegraphics[width=3.6in,height=2.6in]{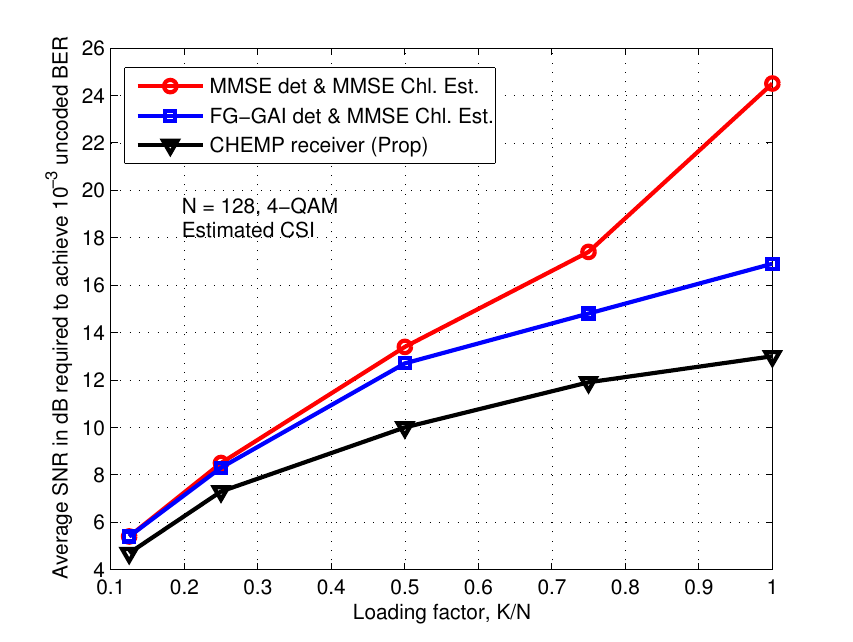} 
\vspace{-4mm}
\caption{
Comparison between the average SNR required to achieve an 
uncoded BER of $10^{-3}$ in 1) proposed CHEMP receiver, 2) MMSE detector 
with MMSE channel estimate, and 3) FG-GAI detector in \cite{jstsp} with 
MMSE channel estimate, at different loading factors with $N=128$, 4-QAM. 
}
\vspace{-4mm}
\label{simmpmdce_lf} 
\end{figure}

Figure \ref{simmpmdce_K} shows the performance of the CHEMP receiver 
and MMSE detector with MMSE channel estimate for different 
number of users ($K=16,32,64,96,128$) and fixed number of BS antennas 
($N=128$). As expected, the performance improves for smaller values of 
$K$.  Also, CHEMP receiver performs better than MMSE detector 
with MMSE channel estimate. In Fig. \ref{simmpmdce_lf}, 
we illustrate
a comparison between the the average SNR required to achieve an uncoded
BER of $10^{-3}$ in 1) proposed CHEMP receiver, 2) MMSE detector with MMSE
channel estimate, and 3) FG-GAI detector in \cite{jstsp} with MMSE channel
estimate, at different loading factors with $N=128$. From this figure,
we observe that the CHEMP receiver outperforms the other two receivers.
For example, the CHEMP receiver outperforms the MMSE detector with 
MMSE channel estimate by about 0.6 dB to 
 11 dB for loading factors in the range of $\alpha=0.125$ to $\alpha=1$.
Likewise, the performance advantage of
the CHEMP receiver over FG-GAI detector with MMSE channel estimate is
about 0.6 dB to  4 dB for loading factors in the range of 
$\alpha=0.125$ to $\alpha=1$.

\subsection{Comparison with SUMIS detector in \cite{sumis}}
\label{sec_sumis}
A subspace marginalization with interference suppression (SUMIS) detector 
has been proposed recently in \cite{sumis}. The SUMIS detector uses the 
ideas of partial marginalization (via a parameter $n_s\in \{1,2,\cdots,K\}$) 
and soft interference suppression. The order of complexity of the SUMIS 
detector is $K^3+2NK+K^2(2n_s^2+6)$ \cite{sumis}. Here, we present a 
performance and complexity comparison between the proposed MPD and the 
SUMIS detector. Figure \ref{fig_sumis1} shows the BER performance of
the proposed MPD and SUMIS detector (with $n_s=3$) for various values of 
$K$ keeping $N$ fixed at 128, 4-QAM, and perfect CSI. For the same system 
parameters, Fig. \ref{fig_sumis2} shows the comparison between the proposed
CHEMP receiver and SUMIS detector with MMSE channel estimate. These figures 
show that the proposed MPD/CHEMP performs better than SUMIS/SUMIS with MMSE 
channel estimate. The proposed detector achieves better performance at less 
complexity than SUMIS detector. This can be observed in Table \ref{comp1} 
which presents the complexities of MPD and SUMIS for different values of 
$N$ and $K$. The complexity advantage of the proposed MPD over SUMIS is 
because MPD needs only matrix multiplication and not matrix inversion, 
whereas SUMIS needs both matrix multiplication and matrix inversion.

\begin{figure}[h]
\includegraphics[width=3.6in,height=2.6in]{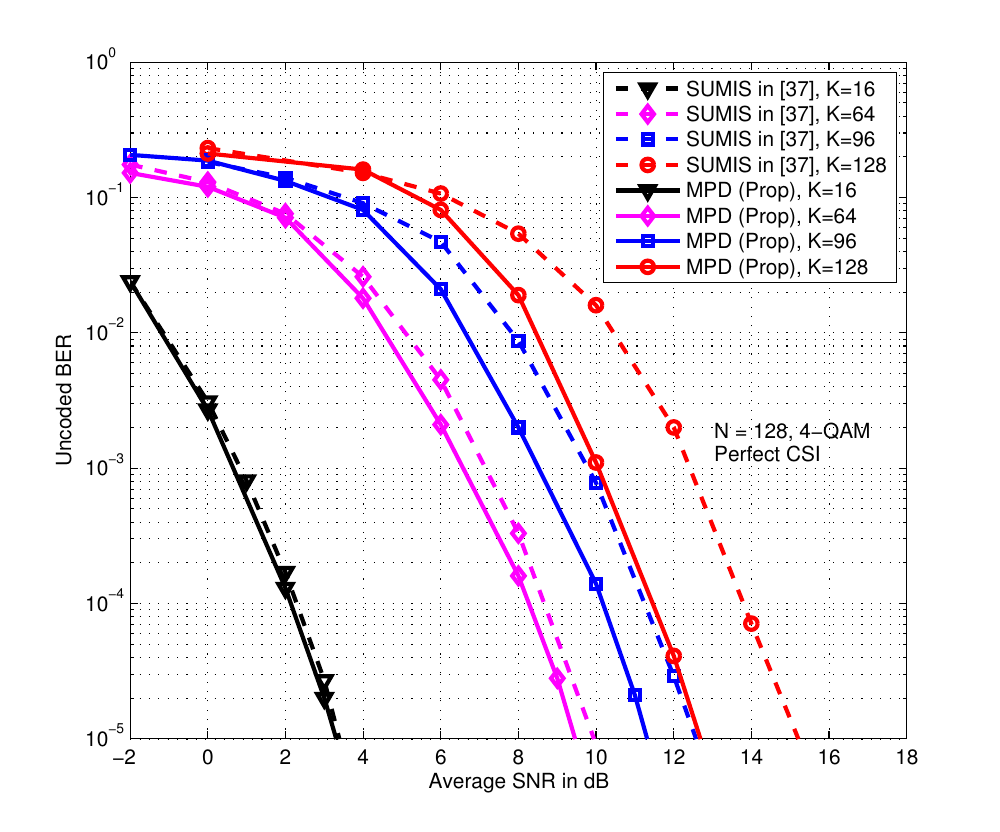}
\vspace{-4mm}
\caption{
BER performance of 1) proposed MPD detector and 2)
SUMIS detector in \cite{sumis} for different values of $K$ \ 
($=16,64,96,128$) for a fixed value of $N$ \ ($=128$), 4-QAM,
perfect CSI.
}
\label{fig_sumis1}
\end{figure}

\begin{figure}[h]
\includegraphics[width=3.6in,height=2.6in]{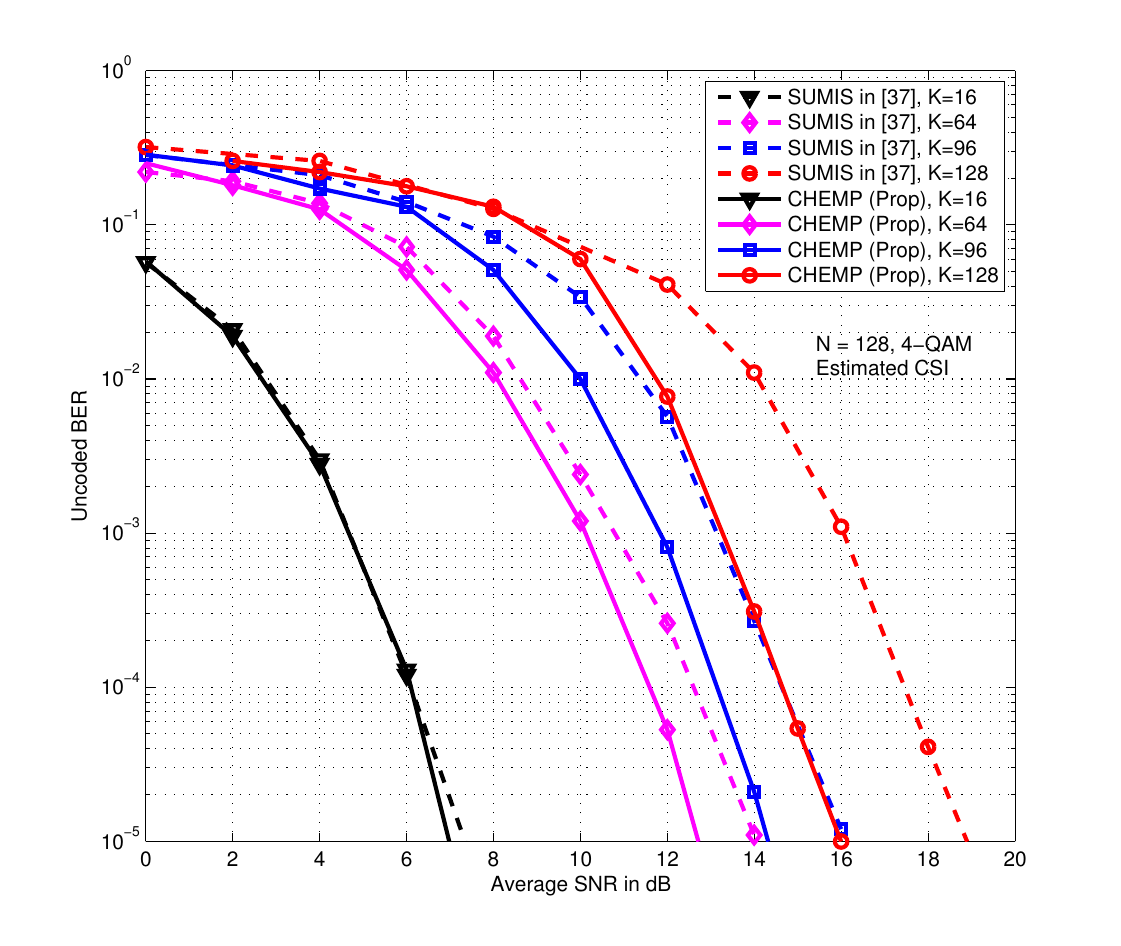}
\vspace{-4mm}
\caption{
BER performance of 1) proposed CHEMP receiver and 2)
SUMIS detector with MMSE channel estimate for different values of $K$ \ 
($=16,32,64,128$) for a fixed value of $N$ \ ($=128$), 4-QAM.}
\label{fig_sumis2}
\end{figure}

\section{Analysis of the proposed CHEMP receiver}
\label{sec4}
In this section, we carry out some analysis of the proposed CHEMP receiver.
The analysis reported in this section has two parts. In the first part, we 
analyze the convergence of the proposed MPD algorithm, and give a sufficient
condition for the algorithm to converge to the correct solution. In the 
second part, we present an analysis of the mean square difference (MSD) of 
the LLRs computed with estimated CSI and perfect CSI for the proposed CHEMP 
receiver as well as the FG-GAI receiver (i.e., FG-GAI 
detector in \cite{jstsp} with MMSE channel estimate). 

\subsection{Analysis of the convergence of MPD algorithm} 
\label{sec4a}
First we state the lemmas that we require to prove results in the later 
parts of this subsection. Let $\mathcal{P}$ denote the set 
$\{\vp \mid \vp\in[0,1]^{2K}\}$.
\begin{lem}
\label{lemma2}
The set $\mathcal{P}$ is a compact and convex set.
\end{lem}
\begin{proof}
Since every element $p_i$ of any $\vp\in\mathcal{P}$ is from the same closed 
compact interval $[0,1]\subset \mathbb{R}$, $\mathcal{P}$ is also a closed 
subset of $\mathbb{R}^{2K}$, and hence $\mathcal{P}$ is also a compact subset.
Let $\vp_1$ and $\vp_2$ be any two elements of $\mathcal{P}$. Then it can be 
seen that for any $\lambda\in[0,1]$,
\begin{equation}
\lambda\vp_1+(1-\lambda)\vp_2\in\mathcal{P}.
\end{equation}
Hence, $\mathcal{P}$ is a convex set. This set $\mathcal{P}$ is the compact
convex subset of $\mathbb{R}^{2K}$ consisting of all probability vectors.
\end{proof}
We define the following variables for convenience:
\begin{eqnarray}
V_i^+(\vp)&\Define&z_i-J_{ii}-\sum\limits_{j=1,j\neq i}^{2K}J_{ij}(2p_j-1), 
\nonumber\\ 
V_i^-(\vp)&\Define&z_i+J_{ii}-\sum\limits_{j=1,j\neq i}^{2K}J_{ij}(2p_j-1),
\label{vs}
\end{eqnarray}
\begin{equation}
A_i^+(\vp)\Define\frac{-1}{2\sigma^2_i}\left(V_i^+(\vp)\right)^2, \quad\quad
A_i^-(\vp)\Define\frac{-1}{2\sigma^2_i}\left(V_i^-(\vp)\right)^2, \nonumber
\end{equation}
\begin{equation}
f_i^+(\vp)\Define\exp\big(A_i^+(\vp)\big), \quad \quad
f_i^-(\vp)\Define\exp\big(A_i^-(\vp)\big),
\end{equation}
where $z_i,J_{ij},\sigma_i$ are constants in $\mathbb{R}$, $\sigma_i>0$ and
$\vp\in\mathcal{P}$.
\begin{lem}
\label{lemma3} 
Let $f(\vp)$ be a function such that if $\vp'=f(\vp)$ then 
$p'_i=f_i(\vp)\Define\frac{f_i^+(\vp)}{f_i^+(\vp)+f_i^-(\vp)}$. Then $f(\vp)$ 
is continuous in $\mathcal{P}$.        
\end{lem}
\begin{proof}
We see that $f:\mathcal{P}\rightarrow\mathcal{P}$. Since $A_i^+(\vp)$ and 
$A_i^-(\vp)$ are polynomial functions in $p_j, j\in\{1,\cdots,2K\}\setminus i$ 
and $\exp(.)$ is a continuous monotone function, $f_i^+(\vp)$ and 
$f_i^-(\vp)$ are continuous functions in $\mathcal{P}$. Since $\vp$ belongs 
to a closed set and $\exp(.)$ is a non-negative function, the term 
\big($f_i^+(\vp)+f_i^-(\vp)\big)$ is always positive. Hence, $f_i(\vp)$ being 
a ratio of two continuous functions with non-vanishing denominator, is also a 
continuous function. This proves that $f(\vp)$ is continuous in $\mathcal{P}$, 
as all its component functions are continuous in $\mathcal{P}$.
\end{proof}

From Lemma \ref{lemma3} we see that $f(\vp)$ is a recursive map that 
represents the proposed MPD algorithm in Section \ref{sec3a}.

\begin{prop}
\label{prop1}
The function $f(\vp)$ defined in Lemma \ref{lemma3} has a fixed point in 
$\mathcal{P}$.
\end{prop}
\begin{proof}
By Lemma \ref{lemma2}, $\mathcal{P}$ is a compact convex set and by Lemma 
\ref{lemma3}, $f(\vp)$ is a continuous function such that 
$f:\mathcal{P}\rightarrow\mathcal{P}$. Hence, by Brouwer's fixed point 
theorem \cite{svbook}, $f(\vp)$ has a fixed point in $\mathcal{P}$.
\end{proof}

Proposition 1 proves that the proposed MPD algorithm has a fixed point.

Now, we give a sufficient condition for the MPD algorithm to converge to the 
correct solution. Since the Gaussian distribution is a symmetric function with 
its positive part being monotone decreasing, we have $p'_i>\frac{1}{2}$ in the 
function $\vp'=f(\vp)$ whenever 
$V_i^+(\vp)^2<V_i^-(\vp)^2$.
Let
\begin{eqnarray}
\hspace{-4mm}
d_i&\Define &V_i^+(\vp)^2-V_i^-(\vp)^2\nonumber\\
&=&\hspace{-3mm}-4J_{ii}\Big[J_{ii}x_i+
\hspace{-3mm}\sum_{j=1,j\neq i}^{2K}J_{ij}(x_j-2p_j+1)+n_i\Big].
\end{eqnarray}
We know that $J_{ii}>0, \forall i$. When $x_i=+1$, $p'_i>\frac{1}{2}$ iff 
$d_i<0$, and $d_i$ will be negative irrespective of $\vp$ iff 
\begin{equation}
J_{ii}+\sum_{j=1,j\neq i}^{2K}J_{ij}(x_j-2p_j+1)+n_i>0.
\label{eqy1}
\end{equation}
Bounding the $J_{ij}(x_j-2p_j+1)$ term on the LHS of (\ref{eqy1}) by
$-2|J_{ij}|$, at high SNRs, we get
\begin{equation}
J_{ii} > 2\sum\limits_{j=1,j\neq i}^{2K}|J_{ij}|.
\label{condition}
\end{equation} 
It can be similarly shown that (\ref{condition}) should be true for $d_i>0$
when $x_i=-1$, irrespective of $\vp$. Thus, when (\ref{condition}) is true
the MPD algorithm has a fixed point that is provably unique and attractive.

When the algorithm starts with an initial vector of $p_i=0.5, \forall i$,
then the condition in (\ref{condition}) can be simplified to
\begin{equation}
J_{ii} > \sum\limits_{j=1,j\neq i}^{2K}|J_{ij}|,
\label{condition2}
\end{equation}
which is nothing but the diagonal dominance condition for the matrix $\mj$, 
and it gives a sufficient condition for the MPD algorithm to converge to 
the correct solution. 
It should be noted that (\ref{condition2}) is a not a necessary condition 
for convergence. From extensive simulations, it has been observed that
the MPD algorithm performs very well for large $N,K$ even when the matrix 
$\mj$ is not diagonally dominant. 

\subsection{Analysis of LLRs in CHEMP and FG-GAI receivers} 
\label{sec4b}
In Fig. \ref{simmpmdce}, we observed that while the performances of MPD and 
FG-GAI under perfect CSI conditions are almost the same, under estimated CSI 
conditions, the CHEMP receiver performs significantly better than FG-GAI with 
MMSE channel estimate. Here, we shall present an LLR analysis that explains 
the reason for this performance advantage of CHEMP receiver under estimated 
CSI conditions.

We note that there are three different LLRs of interest here, which 
we call as Type-1 LLR, Type-2 LLR, and Type-3 LLR. Type-1 LLR is the 
`true' LLR in the `exact' MAP detector. Type-2 LLR is an approximate
LLR in a detector (e.g., MPD, FG-GAI detectors) with perfect CSI. Type-3
LLR is an approximate LLR in a detector with estimated CSI. A comparison
between the Type-I LLR and Type-2 LLR of MPD for large dimensions like
$N=K=128$ is infeasible because of the exponential complexity of the
computation of LLRs in the exact MAP detector. For the purpose of 
analytically reasoning the performance advantage of the CHEMP receiver, 
we use a performance measure which is the mean square difference (MSD) 
between 1) Type-2 and Type-3 LLRs of the MPD detector, and 2) Type-2 and 
Type-3 LLRs of the FG-GAI detector. This MSD measure for a given detector 
can be viewed as an indicator of the relative degradation of the LLR of 
the detector computed under perfect CSI to that computed under estimated 
CSI. In the following, we derive upper bounds on the MSD of LLRs in
CHEMP receiver and FG-GAI with MMSE channel estimate.

The signal vector $\hat{\vz}$ in
the CHEMP receiver given by (\ref{mpmdz}) can be written as 
\begin{eqnarray}
\hat{\vz}&=&\frac{1}{NP}(\my_p^T\mh\vx+\my_p^T\vw)\nonumber \\
&=& \underbrace{\bigg(\mj+\frac{\mw_p^T\mh}{NP}\bigg)}_{\Define \widetilde{\mj}}\vx 
+ \underbrace{\bigg(\frac{\mh}{N}+\frac{\mw_p}{NP}\bigg)^T\vw}_{\Define \ \tvw}
   \nonumber \\
&=&\tmj\vx+\tvw.
\label{eqx1}
\end{eqnarray}
Likewise, the matrix $\mje$ in the CHEMP receiver given by (\ref{mpmdj}) 
can be written as 

\vspace{-2mm}
{\small
\begin{eqnarray}
\hat{\bf J} & \hspace{-2mm} \Define & \hspace{-2mm} \frac{(P\mh+\mw_p)^T(P\mh+\mw_p)}{NP^2}-\frac{\sigma^2_v}{P^2}\mi_K
\nonumber \\
&\hspace{-2mm} = & \hspace{-2mm} \bigg(\mj+\frac{\mw_p^T\mh}{NP}\bigg) + \underbrace{\frac{1}{NP}\bigg(\mh^T\mw_p+
\frac{\mw_p^T\mw_p}{P}\bigg)-\frac{\sigma^2_v}{P^2}\mi_K}_{\Define \ \tmj'}
\nonumber\\
&\hspace{-2mm}=& \hspace{-2mm}\tmj+\tmj'.
\label{eqx2}
\end{eqnarray}
}

\vspace{-4mm}
Note that, as per (\ref{eqx1}), the detection of ${\bf x}$ requires an 
estimate of $\tmj$. But the CHEMP receiver uses $\mje$ instead. This, 
as per (\ref{eqx2}), amounts to using an estimate of $\tmj$ with an
estimation error of $\tmj'$.

Assume $N$ and $K$ are large and all the transmitted bits are i.i.d. 
Let $\delta$ prefixed to a variable denote the difference between the
variable computed under estimated CSI (i.e., using $\mje$ and $\hat\vz$) 
and perfect CSI (i.e., using $\mj$ and $\vz$). For example, 
$\delta\mu_i = \hat{\mu}_i-\mu_i$, where $\hat{\mu}_i$ is obtained by 
substituting $\mje$ in place of $\mj$ in (\ref{eq_mu}).
Likewise, $\delta L_i = \hat{L}_i-L_i$, where $\hat{L}_i$ obtained by
substituting $\mje$ and $\hat\vz$ in place of $\mj$ and $\vz$,
respectively, in (\ref{llreq}).

Now, from (\ref{llreq}), we can write the LLR computed by the CHEMP 
receiver as 
\begin{eqnarray}
\hat{L}_i&=&\frac{2\tj_{ii}+2\tj'_{ii}}{\sigma^2_i+\delta
\sigma^2_i}(\hat z_i-\mu_i-\delta \mu_i). 
\end{eqnarray}
Now, $\delta L_i$ is bounded above as
\begin{eqnarray}
\delta L_i&\leq&\frac{2 \tj'_{ii}(\hat z_i-\mu_i-\delta \mu_i)-2\tj_{ii}\delta
\mu_i}{\sigma^2_i}. 
\label{mpa_e}
\end{eqnarray}

By Lemma \ref{lemma1}, we can write the following: 
\begin{eqnarray}
\tj'_{ij|i\neq
j}&\sim&\mathcal{N}\bigg(0,\frac{\sigma^4_v}{NP^4}+\frac{\sigma^2_v}{2NP^2}\bigg), \\
\tj'_{ii}&\sim&\mathcal{N}\bigg(0,\frac{2\sigma^4_v}{NP^4}+\frac{\sigma^2_v}{2NP^2}\bigg), \\
\delta \mu_i&\sim&\mathcal{N}\bigg(0,\frac{\sigma^4_v}{P^4}+\frac{\sigma^2_v}{2P^2}\bigg).
\end{eqnarray}
Without loss of generality, we can assume $P=1$. Therefore, $\E(\delta L_i)=0$, 
and

\vspace{-4mm}
{\small
\begin{eqnarray} 
\E(\delta L_i^2)&\leq&\frac{\sigma^2_v}{\sigma^4_i}\bigg\{\alpha\bigg(\sigma^2_v+
\frac{1}{2}\bigg)+\bigg(\alpha(\sigma^4_v+\frac{\sigma^2_v}{2})+(z_i-\mu_i)^2\bigg)
\nonumber\\
&&\hspace{8mm}.\bigg(\frac{8\sigma^2_v}{N}+\frac{2}{N}\bigg) \bigg\}.
\label{mpa_ub}
\end{eqnarray}
}

\vspace{-4mm}
Note that $\E(\delta L_i^2)$ is the MSD between the Type-2 and Type-3 
LLRs of the MPD.

Next, we do a similar analysis of the MSD of LLRs for the FG-GAI detector. 
Using the definition of the LLRs $\Lambda_i^k$ in the FG-GAI detector as 
given in \cite{jstsp}, the difference in LLR in FG-GAI computed
with MMSE channel estimate and that computed with perfect CSI is bounded 
above as
\begin{equation}
\delta\Lambda_i^k\leq\frac{4H'_{ik}(y_i-\mu_{ik}-\delta
\mu_{ik})-4H_{ik}\delta \mu_{ik}}{\sigma^2_{ik}}, 
\label{fggai_e}
\end{equation}
where the terms $\mu_{ik}$ and $\sigma^2_{ik}$ are as defined in 
\cite{jstsp}, $H'_{ij}$ is the error in estimating $H_{ij}$, and, as 
defined before, $\delta$ prefixed to a variable denotes the difference 
between that variable computed under estimated CSI and perfect CSI. 
The error in the MMSE channel estimate in the FG-GAI receiver is 
\begin{equation}
H'_{ij}=\frac{W_{ij}P-H_{ij}\sigma^2_n}{P^2+\sigma^2_n},
\end{equation}
where $W_{ij}$ is the $(i,j)$th element in matrix ${\bf W}_P$.
The statistics of $H'_{ij}$ are computed by using Lemma \ref{lemma1} 
as follows:
\begin{equation}
\E(H'_{ij})=0, \quad \quad \sigma^2_e\Define\E({H'_{ij}}^{2})=\frac{\sigma^2_n(P^2+\frac{\sigma^2_n}{2})}{(P^2+\sigma^2_n)^2}.
\end{equation}
Without loss of generality, assume $P=1$ and $\alpha=1$.
Now, we have $H'_{ij}\sim\mathcal{N}(0,\sigma^2_e)$ and 
$\delta \mu_{ij}\sim\mathcal{N}(0,N\sigma^2_e)$. By Lemma \ref{lemma1},
we have $\E(\delta\Lambda_i^j=0)$, and
\begin{equation} 
\E((\delta\Lambda_i^j)^2)\leq\frac{16\sigma^2_e}{\sigma^4_{ij}}\bigg(
N\Big(\sigma^2_e+\frac{1}{2}\Big)+(y_i-\mu_{ij})^2\bigg).
\label{fggai_ub}
\end{equation}
The probability of the $i$th symbol is computed using the LLR value 
$L_i^F\Define\sum_{l\neq i}^N\Lambda_l^j$. Therefore, 
$\delta L_i^F=\sum_{l\neq i}^N\delta\Lambda_l^j$, \, $\E(\delta L_i^F)=0$, 
and $\E((\delta L_i^F)^2) = (N-1)\E(\delta\Lambda_i^j)^2$.
It is noted that $\E((\delta L_i^F)^2)$ is the MSD between the Type-2 and 
Type-3 LLRs of the FG-GAI detector.

It can be seen from (\ref{mpa_ub}) and (\ref{fggai_ub}) that the MSD of 
the computed LLR values in each iteration is less in the CHEMP receiver 
compared to that in the FG-GAI receiver. This is further verified by 
simulation in Fig. \ref{simmpmd_mse}, where it can be observed 
that the simulated MSD of the LLRs in the CHEMP receiver is less compared 
to that in the FG-GAI receiver. This makes the proposed CHEMP receiver 
robust to channel estimation errors when compared to the FG-GAI receiver.

\begin{figure}
\includegraphics[width=3.5in,height=2.5in]{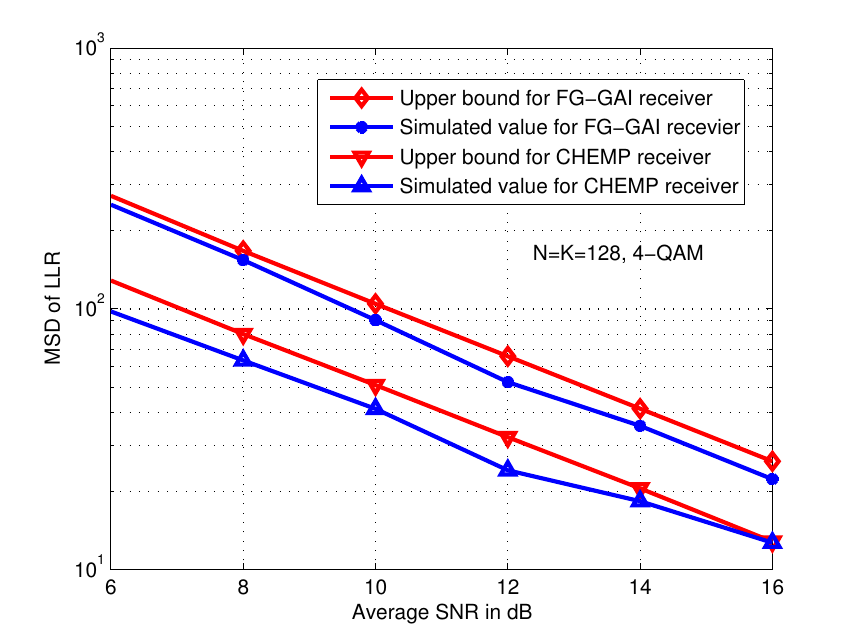} 
\caption{MSD of LLRs in FG-GAI 
and CHEMP receivers for $N=K=128$, 4-QAM.}
\label{simmpmd_mse} 
\end{figure}

\section{Extension to higher-order QAM}
\label{subsec_new}
In this section, we extend the MPD algorithm to higher-order QAM. For
$M$-QAM alphabets, the elements of $\vx$ in (\ref{sys}) belong to the 
underlying PAM alphabet; for example, when the transmitted symbols 
are from 16-QAM alphabet, the elements of $\vx$ are 4-PAM symbols. 
In such a scenario, we compute symbol-wise probability messages in 
the MPD algorithm. Specifically, in each iteration, for each element 
in $\vx$, we compute the probability masses for all symbols in 
$\mathbb B$ as follows.
The means are computed as
\begin{eqnarray}
\mu_i&=&\sum_{j=1,j\neq i}^{2K}J_{ij}\E (x_j)\nonumber\\
&=&\sum_{j=1,j\neq i}^{2K}J_{ij}\sum_{\forall s \in{\mathbb B}}s \ p_j(s).
\end{eqnarray}
The variances are computed as
\begin{eqnarray}
\hspace{-5mm}
\sigma_i^2&=&\sum_{j=1,j\neq i}^{2K}J_{ij}^2\text{Var} (x_j)+\sigma^2_v\nonumber\\
&=&\sum_{j=1,j\neq i}^{2K}J_{ij}^2\Big(\sum_{\forall s\in{\mathbb B}} s^2 p_j(s)-\E(x_j)^2\Big)+\sigma_v^2,
\end{eqnarray}
where $\sigma_v^2$ is as defined in Section \ref{sec3a}.
The probability of $x_i$ being $s \in{\mathbb B}$ is computed as
\begin{equation}
p_i(s)\propto \exp\Big(\frac{-1}{2\sigma^2_i}(z_i-\mu_i-J_{ii}s)^2\Big).
\end{equation}
Finally, the bit probabilities are obtained as
\begin{eqnarray}
\Pr(b_i^p=1)=\sum_{\forall s\in\mathbb{B} : \ p\mbox{\scriptsize{th bit in}}
\hspace{0.5mm} s\hspace{0.5mm} \mbox{\scriptsize{is}} \hspace{0.5mm} 1}p_{i}(s),
\end{eqnarray}
where $b_i^p$ is the $p$th bit in the $i$th user's symbol, which is detected
as 1 if $\Pr(b_i^p=1)\geq 0.5$ and 0 otherwise. It can be
noted that the message passed by each node is a vector of length 
$|\mathbb B|$.

{\em Complexity}:
The complexity of computation of $\vz$ and $\mj$ are $O(NK)$ and $O(NK^2)$,
respectively. The complexity of computing the messages is $O(\sqrt{M}K^2)$ 
for a square $M$-QAM constellation. This is due to the vector nature of the 
messages for $M$-QAM alphabet as opposed to the scalar messages for 
$\{\pm 1\}$ alphabet. In Table \ref{comp2}, we present the complexity for 
16-QAM (in number of real operations) for the proposed MPD, MMSE detector 
and SUMIS detector with $n_s=3$. It can be seen that the complexity of
the proposed MPD is comparable to/less than MMSE complexity and is less 
than SUMIS complexity. In addition, the performance of MPD is better
than those of MMSE and SUMIS detectors as illustrated below.

\begin{table}[h]
\centering
\begin{tabular}{|c||c|c|c|}
\hline
&\multicolumn{3}{|c|}{Complexity in number of real operations $\times 10^{6}$}\\
&\multicolumn{3}{|c|}{$N=128$}\\
\cline{2-4} $K$ & MMSE & MPD & SUMIS\\
 & & (prop) & in \cite{sumis} \\
\hline\hline
16&        0.177& 0.240& 0.483\\ \hline
32&        0.748& 0.964& 1.737\\ \hline
64&        3.593& 3.861& 7.538\\ \hline
96&        9.584& 8.692& 19.368\\ \hline
128&       19.770& 15.456& 39.194\\ \hline
\end{tabular}
\caption{Comparison between the complexities (in number of real operations)
of the proposed MPD, MMSE detection, and SUMIS detection with $n_s=3$ for
16-QAM.}
\label{comp2}
\end{table}

{\em Performance}:
In Fig. \ref{mpd16qam}, we present a comparison between the BER performances 
of the proposed MPD, MMSE detection, and SUMIS detection with $n_s=3$, for
$N=128$, $K=16,32,64$, and 16-QAM. A similar comparison between the proposed
CHEMP receiver, and the MMSE and SUMIS detectors with MMSE channel estimate
is presented in Fig. \ref{chemp16qam}. From these figures, we can see that 
the proposed MPD outperforms the MMSE and SUMIS detectors under perfect
CSI and estimated CSI conditions. 

\begin{figure}[h]
\includegraphics[width=3.6in,height=2.6in]{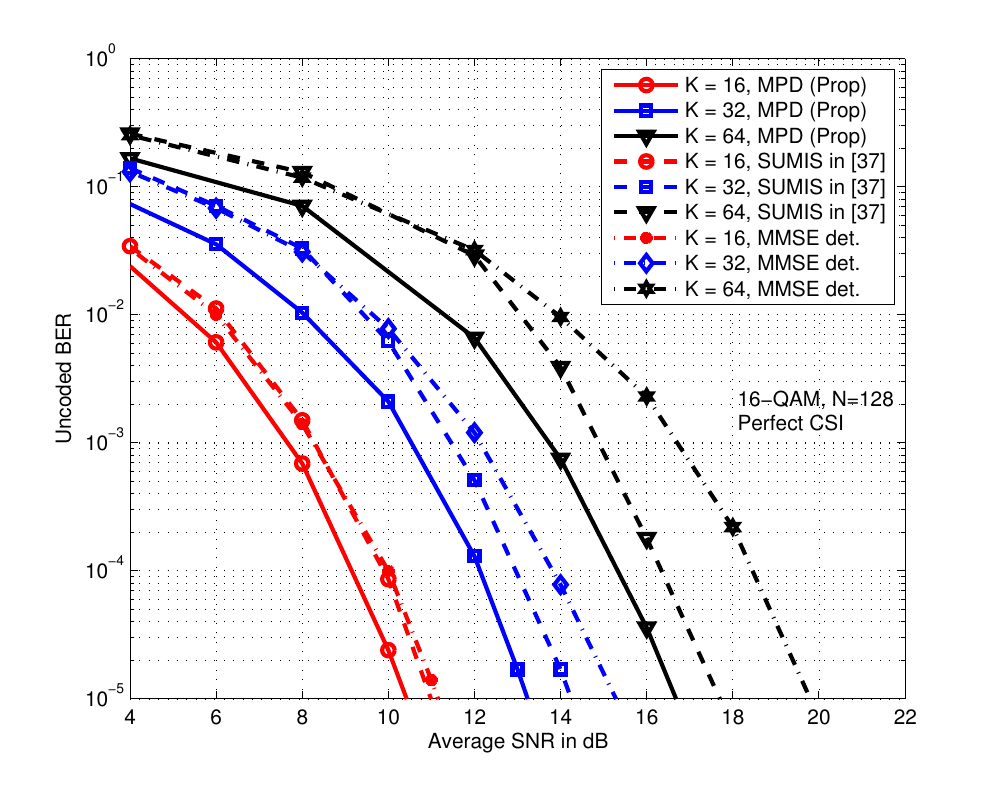}
\vspace{-4mm}
\caption{
Comparison of uncoded BER performance of the proposed MPD,
MMSE detector and SUMIS detector in \cite{sumis} with $n_s=3$
for 16-QAM, $N=128, K=16,32,64$.}
\label{mpd16qam}
\end{figure}

\begin{figure}[h]
\includegraphics[width=3.6in,height=2.6in]{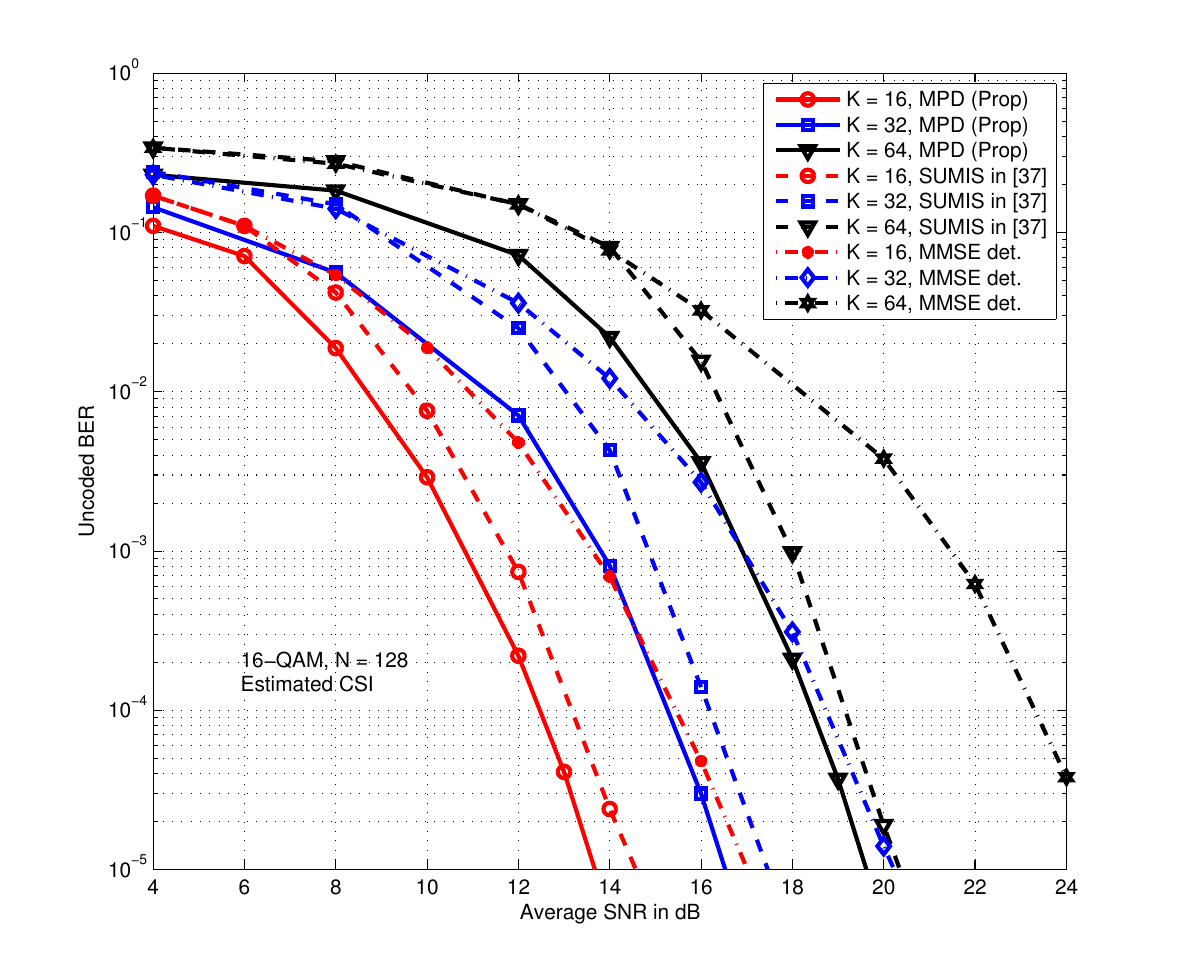}
\vspace{-4mm}
\caption{
Comparison of uncoded BER performance of the proposed CHEMP receiver, 
MMSE and SUMIS detectors with MMSE channel estimate for 16-QAM, $N=128, 
K=16,32,64$.}
\label{chemp16qam}
\end{figure}

\section{Design of LDPC codes for CHEMP receiver}
\label{sec5}
Since both the proposed CHEMP receiver and the LDPC decoder employ message 
passing, a detection-decoding approach based on message passing on a joint 
graph can be natural. In this section, we present a joint graph for the LDPC 
coded system model. We perform MPD and LDPC decoding by passing messages on 
the joint graph. We design optimized irregular LDPC codes specific to the 
considered large MIMO channel and the CHEMP receiver through EXIT chart 
matching. We also present the coded BER performance of the LDPC codes thus 
obtained.

When the detection and decoding operations are performed jointly, the 
receiver starts the detection-decoding process after receiving $n$ coded 
bits. In the joint 
detection-decoding approach, we marginalize the joint probability of the 
received coded symbols. The objective is to compute
\begin{eqnarray}
\Pr ({\bf x}\mid C,{\bf y}) & \propto & \Pr({\bf x},C,{\bf y}) \nonumber \\
 & = & \Pr(C\mid\bf{x})\Pr({\bf y \mid x})\Pr({\bf x}), 
\label{prob-eq}
\end{eqnarray}
where
\begin{equation}
\Pr(C\mid {\bf x})=\prod_{l=1}^{n-k}\Pr(C_{l}\mid {\bf x}),
\end{equation}
$C_l$ is the event of the $l$th check equation of the LDPC code being 
satisfied, and $C$ is the event of all $n-k$ check equations of the LDPC 
code being satisfied. We formulate a graph whose joint probability 
factorizes according to (\ref{prob-eq}), and that upon marginalization 
gives the probability of the transmitted symbols. 

\begin{figure*}
\hspace{2mm}
\includegraphics[width=7.0in,height=2in]{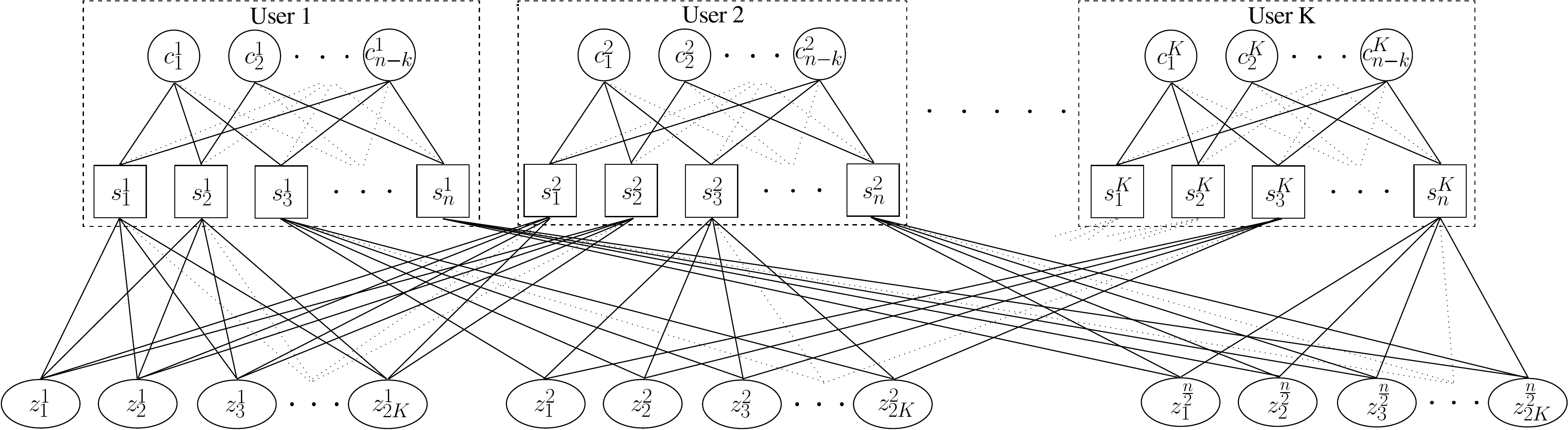} 
\caption{The joint graph of the LDPC coded large-scale MIMO system.}
\label{mpmdldpc_graph} 
\end{figure*}

\subsection{Joint detector and decoder}
\label{sec5a}
Figure \ref{mpmdldpc_graph} shows the joint graph for the LDPC coded 
large-scale MIMO system with 4-QAM. The joint graph consists of three 
sets of nodes, 
namely, variable nodes set, observation nodes set, and check nodes set. 
The $nK$ observation nodes correspond to the elements of the $\vz$ vectors, 
the $nK$ variable nodes correspond to the transmitted coded symbols over 
$\frac{n}{2}$ channel uses, and $(n-k)K$ check nodes correspond to the 
check equations of the LDPC code (see Fig. \ref{mpmdldpc_graph}). 

Let $i\in\{1,\cdots,2K\}$, $j\in\{1,\cdots,K\}$, $m\in\{1,\cdots,n\}$, 
$m'\in\{1,\cdots,\frac{n}{2}\}$, and $l\in\{1,\cdots,n-k\}$. Now, the 
different messages passed over the graph are:
\begin{itemize}
\item {\em Observation node $z_i^{m'}$ to variable node $s_m^j$}:\\
These messages correspond to the probabilities $\Pr(x_i^{m'}=+1)$, the 
probability of the $i$th bit transmitted at the $m'=\lceil\frac{m}{2}\rceil$th 
channel use, i.e., for a given $m'$, $m\in\{2m'-1,2m'\}$.
\item {\em Variable node $s_m^j$ to check node $c_l^j$}:\\
These messages correspond to the probabilities $\Pr(b_m^j=+1)$, the 
probability of the $m$th bit in the LDPC code block transmitted by the 
$j$th user. $l\in\mathcal{N}(s_m^j)$, where $\mathcal{N}(s_m^j)$ is the 
neighborhood of $s_m^j$, i.e., the set of all check nodes connected to 
$s_m^j$. 
\item {\em Check node $c_l^j$ to variable node $s_m^j$}:\\
These messages correspond to the probabilities 
$\Pr(C_l^j\mid s_r^j, \forall r \in \mathcal{N}(c_l^j)\setminus s_m^j)$,
where $\mathcal{N}(c_l^j)$ is the neighborhood of $c_l^j$, i.e., the set of 
all variable nodes connected to $c_l^j$. This corresponds to the probability 
of the $l$th check equation of the LDPC code block transmitted by the $j$th 
user to be satisfied.
\item {\em Variable node $s_m^j$ to observation node $z_i^{m'}$}:\\
These messages correspond to the probabilities 
$\Pr(x_i^{m'}=+1\mid C_r^j, x_u^m, \, \forall r \in \mathcal{N}(s_m^j),
\ u\in\{1,\cdots,2K\}\setminus i )$,
\end{itemize}
It should be noted that, due to the way messages are defined in the MPD of
the CHEMP receiver, there is no message sent from the observation node 
$z_i^{m'}$ to the variable node $s_{2m'-1}^{i}$ when $1\leq i\leq K$, and
there is no message sent from the observation node 
$z_i^{m'}$ to the variable node $s_{2m'}^{i}$ when $K+1\leq i\leq 2K$.
Similarly, the variable node $s_m^j$ sends no message to any observation 
node except $z_{j}^{m'}$ and $z_{2j}^{m'}$. The iterations are continued 
till all the LDPC check equations are satisfied by the estimated bits 
or a certain number of iterations are completed.

\subsection{Design of LDPC codes for the joint detector-decoder}
\label{sec5a}
We obtain the behavior of the proposed joint detector-decoder through 
EXIT curve analysis \cite{tenbrink}. The EXIT function is $f(I_A)=I_E$, 
where $I_E$ is the average mutual information between the coded bits 
and the extrinsic output for a given value of $I_A$, where $I_A$ is the 
average mutual information between the coded bits and the input a priori 
information. First, we obtain the EXIT curves of the CHEMP receiver and 
combine it with that of the LDPC decoder to obtain the EXIT characteristics 
of the joint detector-decoder.

The EXIT characteristics of the CHEMP receiver is obtained through Monte 
Carlo simulations, as an analytical evaluation is intractable. We combine 
the CHEMP receiver's EXIT curves with those of the LDPC decoder, whose 
EXIT curves have known closed-form expressions \cite{tenbrink1}. 
Figure \ref{mpde} shows the EXIT curves of the proposed MPD detector 
and that of the combination of the MPD detector and the variable nodes 
of the LDPC decoder for 4-QAM, $N=128$ and $K=32,128$. We know 
that to approach the capacity of the channel using LDPC codes, we need to 
match the EXIT curves of the check nodes set and the variable nodes set 
\cite{match}, by finding an appropriate degree distribution of the variable 
nodes and the check nodes that is specific for a channel and receiver. Using 
the evaluated EXIT curves and the method detailed in \cite{comm}, we obtain 
the degree distribution of irregular LDPC codes specific for the large-scale 
MIMO channel and the proposed CHEMP receiver. The LDPC codes thus obtained 
for various system parameter settings are presented in Table \ref{tab1}.

\begin{figure}[h]
\centering
\includegraphics[width=3.35in,height=2.5in]{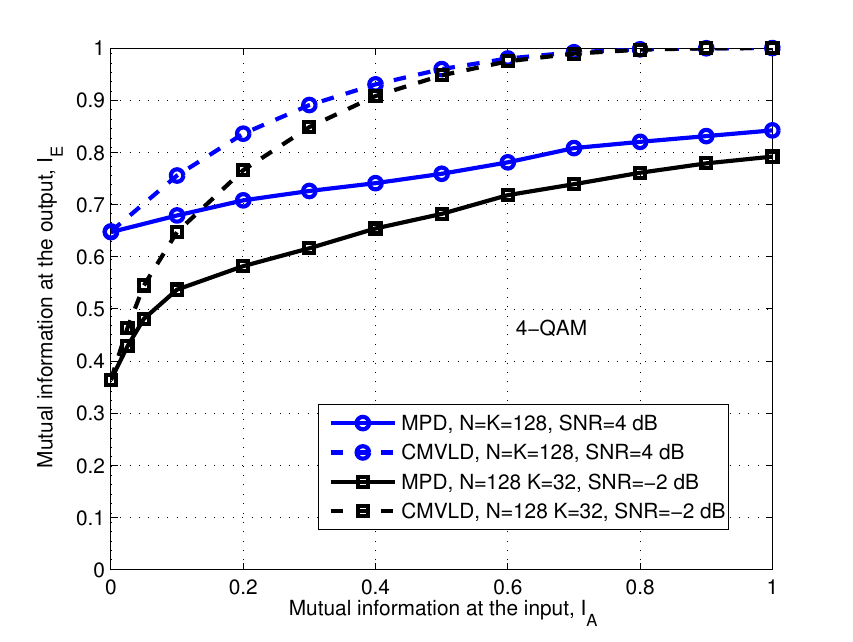}
\caption{EXIT curves of 1) proposed MPD, and 2) combination
of MPD and variable nodes of the LDPC decoder (CMVLD).}
\label{mpde}
\end{figure}

\begin{table}
\centering
\begin{tabular}{|c||c|c|}
\hline
Parameters & ($d_v$, $p_v$) & ($d_c$, $p_c$)\\
\hline \hline
$N=128$, & (2,0.3723), (4, 0.2798), & (6, 0.7067), (12, 0.2531),\\
$\alpha=1$ & (5, 0.2254), (8,0.1152), & (18, 0.0402)\\ 
&(12, 0.0073) &\\ \hline
$N=128$ & (2,0.5715), (4,0.3132), & (4, 0.7045), (8, 0.091) \\
$\alpha=0.5$ & (5, 0.1061), (8, 0.0091) & (12, 0.2045)\\ \hline
$N=128$ & (2,0.4794), (4,0.4201), & (6, 0.7599), (12, 0.1003) \\
$\alpha=0.125$ & (8, 0.0309), (16, 0.0696) & (16, 0.1398)\\
\hline
\end{tabular}
\caption{Degree profiles of optimized rate-1/2 LDPC codes for different 
large MIMO configurations. $p_v$, $p_c$: fraction of variable nodes of 
degree $d_v$ and check nodes of degree $d_c$. 
}
\label{tab1}
\vspace{-0mm}
\end{table}

\subsection{Coded BER performance}
\label{sec5c}
We evaluated the coded BER performance of the joint detector-decoder by 
combining the CHEMP receiver and the LDPC decoder, for $N=128$ and 
$K=16,32,64,96,128$. Figure \ref{simmpmdldpc_N} shows the coded BER 
performance of the optimized LDPC codes for the cases with 1) perfect 
channel knowledge and 2) estimated channel knowledge (i.e., estimated 
$\mh^T\mh$), for $N=K=128$. The minimum SNR required to 
achieve capacity is also marked. The rate of the LDPC code is 
1/2 and the LDPC code block length is $n=$4000. It can be 
seen that the optimized LDPC code performs close to within about 3 dB 
from capacity. We also compare the performance of the optimized codes 
with that of an off-the-shelf irregular LDPC code from \cite{shelf}. 
From Fig. \ref{simmpmdldpc_N}, we can see that the optimized LDPC code 
with perfect channel knowledge performs better than the off-the-shelf 
LDPC code by about 1.2 dB at $10^{-5}$ coded BER. Likewise, the optimized 
LDPC code with estimated channel knowledge outperforms the off-the-shelf 
LDPC code by about 0.8 dB. 

In Fig. \ref{simmpmdldpc_lf}, we plot the average SNRs required to achieve 
a coded BER of $10^{-4}$ by the optimized LDPC codes with estimated channel 
knowledge and perfect channel knowledge, as a function of the system loading 
factor $\alpha$. From Fig. \ref{simmpmdldpc_lf}, we observe that the 
optimized LDPC code with perfect channel knowledge performs better than the 
off-the-shelf LDPC code in \cite{shelf} by about $1.2$ dB at $\alpha=1$, and 
$0.3$ dB at $\alpha=0.125$. Likewise, the optimized LDPC code with the 
estimated channel outperforms the off-the-shelf LDPC code by about $0.7$ dB 
at $\alpha=1$, and $0.5$ dB at $\alpha=0.125$. This performance improvement 
is due to the LDPC code optimization through EXIT curve matching and joint
detection-decoding.

In Fig. \ref{fig_ldpc_new}, we show a performance comparison between the 
proposed optimized code and the codes in \cite{paul} and in the WiMax 
standard \cite{wimax}, in a system with $N=K=128$, 4-QAM, $n=11520$, 
rate-1/2, and perfect CSI. At a block length of $n=11520$, the proposed 
optimized code is found to perform close to within about 2.2 dB from 
capacity. Also, the optimized code is found to perform better than the 
codes in \cite{paul} and \cite{wimax} by about 2 dB and 2.5 dB, 
respectively, at $10^{-5}$ coded BER.

\begin{figure}
\includegraphics[width=3.6in,height=2.6in]{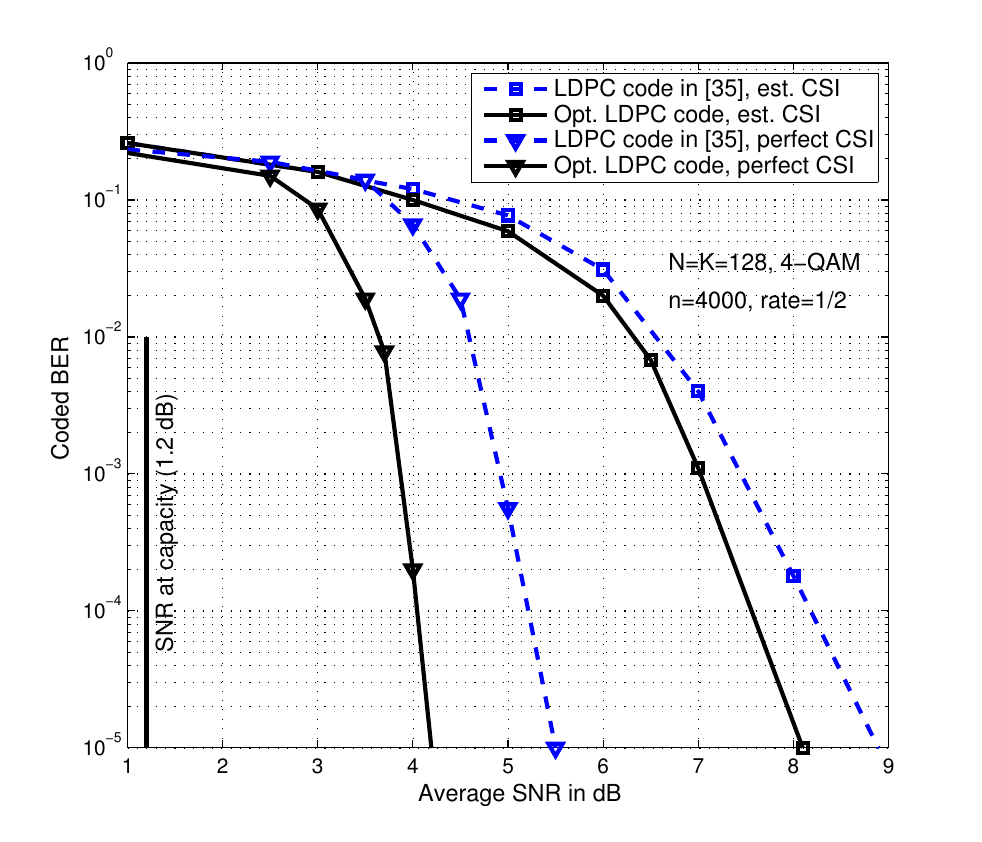} 
\vspace{-4mm}
\caption{Coded BER performance of the irregular LDPC codes 
optimized for the joint detector-decoder with 1) perfect channel knowledge 
and 2) estimated channel knowledge (i.e., estimated $\mh^T\mh$), for 
$N=K=128$, 4-QAM, $n=4000$, rate-1/2.}
\label{simmpmdldpc_N} 
\end{figure}

\begin{figure}
\includegraphics[width=3.6in,height=2.6in]{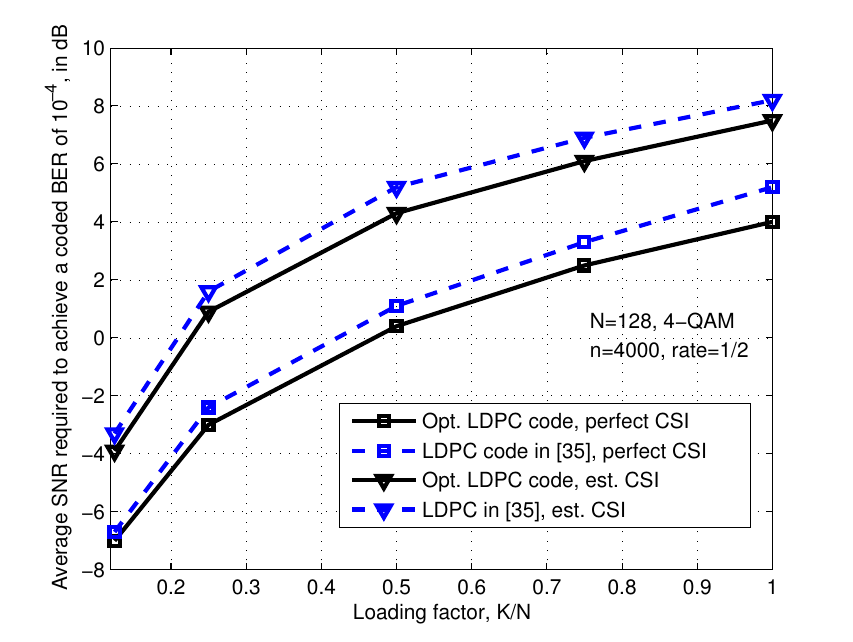} 
\vspace{-4mm}
\caption{Comparison of the average SNR required to achieve a coded BER of
$10^{-4}$ by the joint detector-decoder with 1) perfect channel knowledge
and 2) estimated channel knowledge (i.e., estimated $\mh^T\mh$), for 
various loading factors with $N=128$, 4-QAM, $n=4000$, rate-1/2.}
\label{simmpmdldpc_lf} 
\end{figure}

\begin{figure}
\centering
\includegraphics[width=3.6in,height=2.6in]{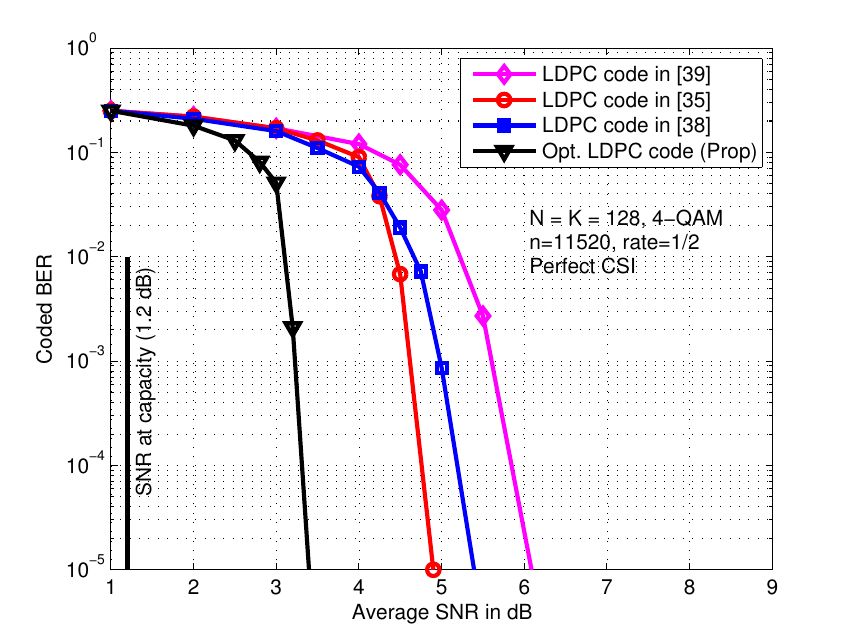} 
\vspace{-4mm}
\caption{Coded BER performance comparison between the optimized 
LDPC code and other LDPC codes in \cite{paul} and in WiMax standard 
\cite{wimax}. $N=K=128$, 4-QAM, $n=11520$, rate-1/2, perfect CSI.}
\label{fig_ldpc_new} 
\end{figure}

\section{Conclusions}
\label{sec6}
We proposed a promising message passing based receiver (referred to as the
`CHEMP receiver') for low complexity detection and channel estimation in 
large-scale MIMO systems. The proposed CHEMP receiver is simple and novel 
(leading to low complexity), yet very effective in large dimensions (leading 
to near-optimal performance). The key idea is a novel way of exploiting the 
channel hardening effect that happens in large MIMO channels. Specifically, 
the receiver worked with approximations to the off-diagonal terms of the 
$\mh^T\mh$ matrix, and directly obtained and used an estimate of $\mh^T\mh$ 
(instead of an estimate of $\mh$). For the considered large-scale MIMO 
settings, the proposed message passing detection algorithm has almost the 
same or less complexity compared to MMSE detection complexity (since the 
proposed detection algorithm does not need a matrix inversion). Yet, it 
could achieve much better performance compared to MMSE detection performance. 
The proposed CHEMP receiver outperformed MMSE and other message passing 
receivers using an MMSE estimate of $\mh$. We presented an analysis of the 
convergence of the proposed detection algorithm and a mean square difference
analysis of the LLRs in proposed receiver with perfect and estimated CSI. 
The irregular LDPC codes obtained for the 
considered large MIMO channel and the proposed CHEMP receiver through EXIT 
chart matching achieved better coded BER performance compared to off-the-shelf 
irregular LDPC codes. Stronger conditions for convergence 
compared to the condition in (\ref{condition}) and convergence analysis 
for the case of estimated channel knowledge are potential topics for 
future research. Extension of the proposed receiver approach to 
frequency-selective channels can also be carried out as future extension 
to this work. 

\bibliographystyle{ieeetr} 

\end{document}